\documentclass[11pt,a4paper]{article}
\usepackage{amssymb}
\usepackage{amsmath}
\usepackage{amsfonts}
\usepackage{bbm}
\usepackage{amsthm}
\usepackage{mathrsfs}
\usepackage{hyperref}
\usepackage{color}
\usepackage[margin=2.41cm]{geometry}
\usepackage[all,cmtip]{xy}
\usepackage[utf8]{inputenc}
\usepackage{graphicx}
\usepackage{varwidth}
\usepackage{comment}
\usepackage{enumitem}

\usepackage{upgreek}
\usepackage{rotating}

\usepackage{tikz}
\usetikzlibrary{shapes.geometric}

\usepackage[most]{tcolorbox}

\definecolor{darkred}{rgb}{0.8,0.1,0.1}
\hypersetup{
	colorlinks=false,         
	linkcolor=darkred,
	citecolor=blue,
}

\theoremstyle{plain}
\newtheorem{theo}{Theorem}[section]
\newtheorem{lem}[theo]{Lemma}
\newtheorem{propo}[theo]{Proposition}

\theoremstyle{definition}
\newtheorem{defi}[theo]{Definition}

\newtheorem{exx}[theo]{Example}
\newenvironment{ex}
{\pushQED{\qed}\exx}
{\popQED\endexx}

\newtheorem{remm}[theo]{Remark}
\newenvironment{rem}
{\pushQED{\qed}\remm}
{\popQED\endremm}

\numberwithin{equation}{section}

\def\nn{\nonumber}

\def\bbK{\mathbb{K}}
\def\bbR{\mathbb{R}}
\def\bbC{\mathbb{C}}

\def\bbZ{\mathbb{Z}}

\def\ii{{\,{\rm i}\,}}

\def\Hom{\mathrm{Hom}}

\def\Imm{\mathrm{Im}}
\def\Ker{\mathrm{Ker}}

\def\Sym{\mathrm{Sym}}

\def\id{\mathrm{id}}
\def\supp{\mathrm{supp}}

\def\dd{\mathrm{d}}
\def\vol{\mathrm{vol}}

\def\cc{\mathrm{c}}

\def\CS{\mathrm{CS}}
\def\MW{\mathrm{MW}}

\def\1{I}
\def\oone{\mathbbm{1}}
\def\op{\mathrm{op}}

\def\BV{\mathrm{BV}}

\def\Loc{\mathbf{Loc}}

\def\Vec{\mathbf{Vec}}
\def\Ch{\mathbf{Ch}}

\def\dgAlg{\mathbf{dgAlg}}
\def\dgCAlg{\mathbf{dgCAlg}}

\def\PCh{\mathbf{PoCh}}
\def\VBun{\mathbf{VBun}}

\def\AQFT{\mathbf{AQFT}_m}
\def\tPFA{\mathbf{tPFA}_m}

\def\FFF{\mathfrak{F}}

\def\ttt{\mathfrak{t}}

\def\A{\mathcal{A}}
\def\B{\mathcal{B}}

\def\F{\mathcal{F}}
\def\G{\mathcal{G}}

\def\Q{\mathcal{Q}}

\def\FFFF{\FFF}

\newcommand\und[1]{\underline{#1}}

\newcommand{\pair}[1]{\{\!\{#1\}\!\}}

\def\sk{\vspace{2mm}}

\makeatletter
\let\@fnsymbol\@alph
\makeatother

\title{%
Quantization of Lorentzian free BV theories:\\ factorization algebra vs algebraic quantum field theory
}

\author{%
Marco Benini$^{1,2,a}$, 
Giorgio Musante$^{1,b}$\ and\ 
Alexander Schenkel$^{3,c}$\vspace{4mm}\\
{\small ${}^1$ Dipartimento di Matematica, Dipartimento di Eccellenza 2023-27, Universit\`a di Genova,}\\
{\small Via Dodecaneso 35, 16146 Genova, Italy.}\vspace{2mm}\\
{\small ${}^2$ INFN, Sezione di Genova,}\\
{\small Via Dodecaneso 33, 16146 Genova, Italy.}\vspace{2mm}\\
{\small ${}^3$ School of Mathematical Sciences, University of Nottingham,}\\
{\small University Park, Nottingham NG7 2RD, United Kingdom.}\vspace{4mm}\\
{\small \begin{tabular}{ll}
Email: & ${}^a$~\texttt{marco.benini@unige.it}\\
& ${}^b$~\texttt{musante@dima.unige.it}\\
& ${}^c$~\texttt{alexander.schenkel@nottingham.ac.uk}
\end{tabular}
}
}

\date{February 2024}

\begin{document}

\maketitle

\begin{abstract}
\noindent We construct and compare two alternative quantizations, as a time-orderable prefactorization algebra and as an algebraic quantum field theory valued in cochain complexes, of a natural collection of free BV theories on the category of $m$-dimensional globally hyperbolic Lorentzian manifolds. Our comparison is realized as an explicit isomorphism of time-orderable prefactorization algebras. The key ingredients of our approach are the retarded and advanced Green's homotopies associated with free BV theories, which generalize retarded and advanced Green's operators to cochain complexes of linear differential operators. 
\end{abstract}

\paragraph*{Keywords:} factorization algebras, 
algebraic quantum field theories, 
homological methods in gauge theory, 
globally hyperbolic Lorentzian manifolds, 
Green hyperbolic operators. 
\vspace{-2mm}

\paragraph*{MSC 2020:} 81T70, 81T20, 58J45. 
\vspace{-2mm}

\renewcommand{\baselinestretch}{0.8}\normalsize
\tableofcontents
\renewcommand{\baselinestretch}{1.0}\normalsize

\newpage


\section{\label{sec:intro}Introduction and summary}
Several mathematical axiomatizations of quantum field theory (QFT) 
on Lorentzian manifolds have been proposed in the literature, 
such as algebraic quantum field theories (AQFTs) 
\cite{BFV, FV, operads_aqft} and time-orderable prefactorization algebras 
\cite{FA-vs-AQFT}, i.e.\ a Lorentzian variant of prefactorization algebras 
\cite{CG, CG2}. These two approaches are a priori quite different. 
For instance, while the former emphasizes the algebraic 
structure carried by the quantum observables on each spacetime, 
the latter focuses on their time-ordered products. 
The differences become even more striking when one tries 
to construct simple QFT models, such as the free Klein-Gordon field 
of mass $m \geq 0$: 
while the corresponding time-orderable prefactorization algebra 
is constructed out of the $(-1)$-shifted Poisson structure (antibracket)
$\tau_{(-1)}(\varphi \otimes \varphi^\ddagger) := \int_M \varphi \varphi^\ddagger\, \vol_M$ only, 
the corresponding AQFT relies crucially also on the retarded 
and advanced Green's operators $G_\pm$ for the Klein-Gordon operator $\Box + m^2$
through the unshifted Poisson structure 
$\tau_{(0)}(\varphi_1 \otimes \varphi_2) := \int_M \varphi_1 (G_+ - G_-) \varphi_2\, \vol_M$. 
\sk

Because of these differences it is interesting to compare
time-orderable prefactorization algebras and AQFTs.
This task was undertaken first in a model-based approach 
by \cite{GwilliamRejzner} and then in a model-independent fashion 
by \cite{FA-vs-AQFT}. 
In \cite{GwilliamRejzner} it is shown that the time-orderable
prefactorization algebra and the AQFT of the free
Klein-Gordon field encode equivalent information as a consequence
of the {\it time-slice axiom}, i.e.\ the property that
any spacetime embedding whose image contains a Cauchy surface of the
codomain induces an isomorphism at the level of quantum observables.
(The results of \cite{GwilliamRejzner} can be adapted with minor
modifications to encompass any field theoretic model that is ruled
by a Green hyperbolic operator.)
In \cite{FA-vs-AQFT} a model-independent comparison is developed in the
form of an equivalence (actually isomorphism) between the categories
of time-orderable prefactorization algebras and of AQFTs,
both satisfying the time-slice axiom (and an additional technical
requirement, called additivity, that is fulfilled by many examples).
\sk

Unfortunately, the results of \cite{GwilliamRejzner} and \cite{FA-vs-AQFT} 
do not cover the examples of linear gauge theories. 
On the one hand, the equation of motion of a linear gauge theory 
(with gauge transformations acting non-trivially) must be degenerate. 
In particular, the corresponding linear differential operator is not 
a Green hyperbolic operator, see \cite{Bar} and also Definition~\ref{defi:Gpm}. 
As a consequence, the results of \cite{GwilliamRejzner} 
cannot be applied directly. 
On the other hand, linear gauge theories are most naturally 
encoded by cochain complexes in the spirit of the BV formalism, 
see \cite{CG, CG2, Hol, FR, FR2, BSreview}. In this context 
a weaker version of the time-slice axiom holds, 
where isomorphisms are replaced by quasi-isomorphisms,
see \cite{linYM} and also Examples~\ref{ex:tPFA} and~\ref{ex:AQFT}.
Motivated by this fact, linear gauge theories on Lorentzian manifolds 
are realized by means of time-orderable prefactorization algebras or AQFTs
that take values in the $\infty$-category $\Ch_\bbC$
of cochain complexes with equivalences given by quasi-isomorphisms,
see Definitions~\ref{defi:tPFA} and~\ref{defi:AQFT} 
and also Remark~\ref{rem:hAQFT}. 
\sk

While we are currently not able to upgrade 
the model-independent comparison of \cite{FA-vs-AQFT} 
to the case where the target is the $\infty$-category $\Ch_\bbC$, 
with the present paper we extend the results 
of \cite{GwilliamRejzner} to linear gauge (and also higher gauge) theories. 
The key ingredient to achieve this goal is a 
generalization of Green hyperbolic operators, namely 
the recently developed {\it Green hyperbolic complexes} \cite{GreenHyp}. 
In contrast to Green hyperbolic operators, Green hyperbolic 
complexes cover many important examples of linear gauge theories, 
see \cite{GreenHyp} and also Examples~\ref{ex:P}, \ref{ex:CS} 
and~\ref{ex:MW}. Their key feature is that they admit 
{\it retarded and advanced Green's homotopies} 
$\Lambda_\pm$, generalizing the familiar retarded and advanced Green's 
operators $G_\pm$ for Green hyperbolic operators. 
\sk

The input of our construction is a {\it free BV theory} $(F,Q,(-,-),W)$ 
on an $m$-dimensional oriented and time-oriented globally hyperbolic 
Lorentzian manifold $M \in \Loc_m$, consisting 
of a {\it complex of linear differential operators} $(F,Q)$ 
with a compatible {\it $(-1)$-shifted fiber metric} $(-,-)$ 
and a {\it (formally self-adjoint) Green's witness}, 
see Definitions~\ref{defi:CLDO-fibermetric}, \ref{defi:witness} 
and~\ref{defi:BV}. Let us provide some interpretation of these data 
and some information about the structures that can be defined out of it. 
In the spirit of the BV formalism one may think of 
the graded vector bundle $F$ as encoding both gauge and ghost fields, 
as well as the respective antifields. In the same spirit the differential $Q$, 
which is degree-wise a linear differential operator, encodes 
both the action of gauge transformations and the equation of motion. 
The compatible $(-1)$-shifted fiber metric $(-,-)$ is a suitable
generalization of the more familiar concept of a fiber metric on a
vector bundle. $(-,-)$ is closely related to the antibracket 
from the BV formalism in the sense that, upon integration, it defines 
the {\it $(-1)$-shifted Poisson structure} $\tau_{(-1)}$ 
on the $1$-shift $\FFF_\cc(M)[1] \in \Ch_\bbR$ of the cochain complex 
of compactly supported smooth sections of $(F,Q)$, 
see \eqref{eqn:shiftedPoiss}. Finally, the role of the Green's witness 
$W$ is to give rise to the Green hyperbolic operator $P := Q\, W + W\, Q$, 
which allows one to find particularly simple retarded and advanced Green's 
homotopies $\Lambda_\pm := W\, G_\pm$, where $G_\pm$ denote the retarded and 
advanced Green's operator for $P$. In this sense $W$ ``witnesses'' 
the fact that $(F,Q)$ is a Green hyperbolic complex. 
Taking the difference of $\Lambda_+$ and $\Lambda_-$ defines the {\it retarded-minus-advanced cochain map}
$\Lambda := \Lambda_+ - \Lambda_-$ and taking their average defines 
the {\it Dirac homotopy} $\Lambda_D := \tfrac{1}{2}(\Lambda_+ + \Lambda_-)$, 
which generalize the familiar retarded-minus-advanced $G := G_+ - G_-$ 
and Dirac $G_D := \tfrac{1}{2}(G_+ + G_-)$ propagators. 
In combination with the $(-1)$-shifted fiber metric $(-,-)$, 
$\Lambda$ and $\Lambda_D$ define, upon integration, 
the {\it unshifted Poisson structure} $\tau_{(0)}$ 
and respectively the {\it Dirac pairing} $\tau_D$ on the cochain complex 
$\FFF_\cc(M)[1] \in \Ch_\bbR$ of compactly supported smooth sections, 
see \eqref{eqn:Poiss} and \eqref{eqn:pairDirac}. 
The $(-1)$-shifted Poisson structure $\tau_{(-1)}$ plays a crucial role 
in the first step of our construction 
(quantization as a time-orderable prefactorization algebra), 
the unshifted Poisson structure $\tau_{(0)}$ 
in the second step (quantization as an AQFT) 
and the Dirac pairing $\tau_D$ in the last step (comparison). 
\sk

In the first step, which is carried out in Subsection~\ref{subsec:BV-quant}, 
we construct a time-orderable prefactorization algebra $\F \in \tPFA$
out of a collection $(F_M,Q_M,(-,-)_M,W_M)_{M \in \Loc_m}$ of free BV theories
that is natural with respect to the morphisms $f: M \to N$ in $\Loc_m$
(see Appendix~\ref{app:natural} for the technical details). 
The first part of this construction relies only on the complexes of linear
differential operators $(F_M,Q_M)$ and on the compatible $(-1)$-shifted fiber 
metrics $(-,-)_M$, for all $M \in \Loc_m$. These data 
are used to define the $(-1)$-shifted Poisson structures $\tau_{(-1)}$, 
whose BV quantization provides the time-orderable prefactorization 
algebra $\F$ of interest to us. Explicitly, from $\tau_{(-1)}$ 
we define the {\it BV Laplacian} $\Delta_{\BV}$ 
on the symmetric algebra $\Sym(\FFF_\cc(M)[1]) \in \dgCAlg_\bbC$, 
see \eqref{eqn:BV-Laplacian}, and then we deform the original 
differential $\Q$ to the quantized differential 
$\Q_\hbar := \Q + \ii \hbar\, \Delta_{\BV}$. 
Even though $\Q_\hbar$ is not compatible with the commutative 
multiplication $\mu$ of the symmetric algebra 
$\Sym(\FFF_\cc(M)[1]) \in \dgCAlg_\bbC$, it is compatible with the 
{\it time-ordered products} constructed out of $\mu$, 
see Proposition~\ref{propo:t-ord-prod}. 
Hence, by defining for all $M \in \Loc_m$ the cochain complexes 
$\F(M) := (\Sym(\FFF_\cc(M)[1]), \Q_\hbar) \in \Ch_\bbC$  
that consist of the graded vector space underlying 
$\Sym(\FFF_\cc(M)[1])$ with the quantized differential 
$\Q_\hbar$, we obtain the time-orderable prefactorization algebra 
$\F$ with time-ordered products constructed out of 
the symmetric algebra multiplication $\mu$. 
At this point it is unclear whether $\F$ fulfills the time-slice axiom. 
The Green's witnesses $W_M$ become crucial for this purpose, 
see Theorem~\ref{th:time-slice-cl} and  
Proposition~\ref{propo:tPFA-time-slice}. 
\sk

In the second step, which is carried out in Subsection~\ref{subsec:star-prod}, 
we construct an AQFT $\A \in \AQFT$ out of the same data. 
Explicitly, instead of using $\tau_{(-1)}$ to deform the differential, 
here one uses the unshifted Poisson structure $\tau_{(0)}$ 
to deform the commutative multiplication $\mu$ 
of the symmetric algebra $\Sym(\FFF_\cc(M)[1]) \in \dgCAlg_\bbC$ to the 
(in general non-commutative) {\it Moyal-Weyl star product $\mu_\hbar$}, 
see \eqref{eqn:star-prod}. The deformed multiplication $\mu_\hbar$ is compatible with the original 
differential $\Q$ and with the pushforward of compactly supported 
sections along $\Loc_m$-morphisms. Hence, we obtain the AQFT $\A$ 
by defining for all $M \in \Loc_m$ the differential graded 
algebras $\A(M) := (\Sym(\FFF_\cc(M)[1]),\mu_\hbar,\oone) \in \dgAlg_\bbC$ 
that consist of (the cochain complex underlying) 
$\Sym(\FFF_\cc(M)[1])$ with the Moyal-Weyl star product $\mu_\hbar$ 
and the unit $\oone \in \Sym(\FFF_\cc(M)[1])$,
and extending the pushforward of compactly supported sections. 
\sk

In the last step, which is carried out in Subsection~\ref{subsec:comp}, 
we compare the time-orderable prefactorization algebra 
$\F \in \tPFA$ and the AQFT $\A \in \AQFT$ obtained in the previous steps. Explicitly, 
we construct a comparison isomorphism $T: \F \to \F_\A$ in $\tPFA$
between $\F \in \tPFA$ and the time-orderable prefactorization algebra 
$\F_\A \in \tPFA$ associated with $\A \in \AQFT$, whose time-ordered products 
are constructed out of the Moyal-Weyl star product $\mu_\hbar$. 
($\F_\A$ is just the evaluation on $\A$ of the functor 
$\AQFT \to \tPFA$ from \cite{FA-vs-AQFT}.) 
The comparison isomorphism $T := \exp(\ii \hbar\, \Delta_D)$ 
is defined as the exponential of the {\it Dirac Laplacian} 
$\Delta_D$, which is obtained from the Dirac pairing $\tau_D$, 
see Theorem~\ref{th:iso}. In particular, we show that $T$ intertwines 
the quantized differential $\Q_\hbar$ with the original symmetric algebra differential $\Q$ 
and the time-ordered products constructed out of the original symmetric algebra multiplication $\mu$ 
with those constructed out of the quantized multiplication $\mu_\hbar$. 
\sk

The outline of the rest of the paper is the following. 
Section~\ref{sec:prelim} contains the background material needed later on. 
In particular, Subsection~\ref{subsec:Ch} reviews some basic aspects 
of the theory of cochain complexes $\Ch_\bbK$ over a field $\bbK$ 
of characteristic zero; Subsection~\ref{subsec:pairing} describes the 
extension of (anti-)symmetric pairings $\tau$ of degree $p \in \bbZ$ 
on a cochain complex $V \in \Ch_{\bbK}$ to suitable bi-derivations 
$\pair{-,-}_\tau$ and, in the symmetric case, 
to suitable Laplacians $\Delta_\tau$ 
on the symmetric algebra $\Sym\,V\in\dgCAlg_\bbK$; 
Subsection~\ref{subsec:Green} recalls some relevant concepts 
from Lorentzian geometry and Green hyperbolic operators; 
Subsection~\ref{subsec:AQFTandtPFA} reviews the concepts 
of time-orderable prefactorization algebras and AQFTs valued 
in cochain complexes $\Ch_\bbC$, including the Einstein causality 
and time-slice axioms (the latter in the form of a quasi-isomorphism). 
Section~\ref{sec:Witness} focuses on the concept of a Green's witness 
and on the structures that can be constructed out of it. 
More in detail, Subsection~\ref{subsec:taus} recalls the concepts 
of a complex of linear differential operators $(F,Q)$, 
of a compatible $(-1)$-shifted fiber metric $(-,-)$ and 
of a (formally self-adjoint) Green's witness $W$, 
which together form a free BV theory $(F,Q,(-,-),W)$ 
on $M \in \Loc_m$, and out of these data it constructs 
the $(-1)$-shifted Poisson structure $\tau_{(-1)}$, 
the unshifted Poisson structure $\tau_{(0)}$ and 
the Dirac pairing $\tau_D$; Subsection~\ref{subsec:properties} 
investigates the properties of the structures 
$\tau_{(-1)}$, $\tau_{(0)}$ and $\tau_D$ associated with a natural 
collection $(F_M,Q_M,(-,-)_M,W_M)_{M \in \Loc_m}$ of free BV theories, 
proving in particular classical analogs of the Einstein causality 
and time-slice axioms, see Theorem~\ref{th:time-slice-cl}. 
The core of the paper is Section~\ref{sec:Quant}, which is devoted to 
the construction and comparison of two alternative quantizations 
of a natural collection $(F_M,Q_M,(-,-)_M,W_M)_{M \in \Loc_m}$ 
of free BV theories. The starting point of both quantization schemes is the 
symmetric algebra $\Sym(\FFF_\cc(M)[1]) \in \dgCAlg_\bbC$, 
where $\FFF_\cc(M)[1] \in \Ch_\bbR$ denotes the $1$-shift 
of the cochain complex of compactly supported smooth sections 
of the complex of linear differential operators $(F_M,Q_M)$. 
Subsection~\ref{subsec:BV-quant} quantizes 
$(F_M,Q_M,(-,-)_M,W_M)_{M \in \Loc_m}$ as a time-orderable prefactorization 
algebra $\F \in \tPFA$ by deforming the original differential $\Q$ of 
$\Sym(\FFF_\cc(M)[1])$ 
to the quantized differential $\Q_\hbar := \Q + \ii \hbar\, \Delta_{\BV}$ 
by means of the BV Laplacian $\Delta_\BV$ defined from 
the $(-1)$-shifted Poisson structure $\tau_{(-1)}$; 
Subsection~\ref{subsec:star-prod} quantizes 
$(F_M,Q_M,(-,-)_M,W_M)_{M \in \Loc_m}$ as an AQFT $\A \in \AQFT$ 
by deforming the original commutative multiplication $\mu$ of 
the symmetric algebra $\Sym(\FFF_\cc(M)[1]) \in \dgCAlg_\bbC$ 
to the (in general non-commutative) Moyal-Weyl star product $\mu_\hbar$ 
by means of the bi-derivation $\pair{-,-}_{(0)}$ extending 
the unshifted Poisson structure $\tau_{(0)}$; 
Subsection~\ref{subsec:comp} concludes the paper with
constructing in Theorem~\ref{th:iso} an isomorphism 
$T := \exp(\ii \hbar\, \Delta_D): \F \to \F_\A$ in $\tPFA$, 
where $\Delta_D$ denotes the Dirac Laplacian 
defined from the Dirac pairing $\tau_{D}$. $T$ intertwines 
the quantized differential $\Q_\hbar$ and original 
(i.e.\ constructed out of $\mu$) time-ordered products of $\F \in \tPFA$ 
with the original differential $\Q$ and quantized 
(i.e.\ constructed out of $\mu_\hbar$) time-ordered products 
of the time-orderable prefactorization algebra $\F_\A \in \tPFA$ 
associated with $\A \in \AQFT$ according to \cite{FA-vs-AQFT}. 
Appendix \ref{app:natural} discusses some technical details about 
naturality of vector bundles, fiber metrics and differential operators
which we require to introduce the concept of natural free BV theories.


\section{\label{sec:prelim}Preliminaries}

\subsection{\label{subsec:Ch}Cochain complexes}
We review some basic aspects of the theory of cochain complexes 
to fix our notation and conventions.
More details on the well-known topics recalled here are covered 
by the classical literature, see e.g.\ \cite{Weibel, Hovey}. 
Let us fix a field $\bbK$ of characteristic zero. 
In the main part of this paper $\bbK$ will be either the field
$\bbR$ of real numbers or the field $\bbC$ of complex numbers. 
\begin{defi}
A {\it cochain complex}
$V = (V, Q_V)$ consists of a $\bbZ$-graded $\bbK$-vector space $V=(V^n)_{n\in\bbZ}$ 
and a differential $Q_V$, that is a collection $Q_V = (Q_V^n)_{n \in \bbZ}$ 
of degree increasing $\bbK$-linear maps 
$Q_V^n : V^n \to V^{n+1}$ such that $Q_V^{n+1} Q_V^n = 0$, for all $n \in \bbZ$.
A {\it cochain map} $f : V \to W$ is a family $f = (f^n)_{n \in \bbZ}$ of $\bbK$-linear maps 
$f^n : V^n \to W^n$ that is compatible with the differentials, 
i.e.\ $Q_W^n \, f^n = f^{n+1} \,Q_V^n$, for all $n \in \bbZ$. 
We denote by $\Ch_\bbK$ the category whose objects are cochain complexes 
and whose morphisms are cochain maps.
\end{defi}

The tensor product $V \otimes W \in \Ch_\bbK$ of two cochain complexes $V, W \in \Ch_\bbK$ consists of
\begin{subequations}\label{eqn:tensorproduct}
\begin{flalign}
(V \otimes W)^n := \bigoplus_{p \in \bbZ}{(V^p \otimes W^{n-p})} \quad,
\end{flalign}
for all $n \in \bbZ$, and of the differential $Q_\otimes$ given by the graded Leibniz rule
\begin{flalign}
Q_\otimes (v \otimes w) := Q_V v \otimes w + (-1)^{\lvert v \rvert}\, v \otimes Q_W w \quad,
\end{flalign}
\end{subequations}
for all homogeneous $v \in V$ and $w \in W$, 
where $\lvert - \rvert$ denotes the degree.
The monoidal unit of the tensor product is given by $\bbK \in \Ch_\bbK$, regarded as 
a cochain complex concentrated in degree zero with trivial differential.
The symmetric braiding is given by the cochain maps $\gamma : V \otimes W \to W \otimes V$ in $\Ch_\bbK$ 
that are defined by the Koszul sign rule 
\begin{flalign}
\gamma (v \otimes w) := (-1)^{\lvert v \rvert \, \lvert w \rvert}\, w \otimes v\quad,
\end{flalign}
for all homogeneous $v \in V$ and $w \in W$.
The internal hom $[V, W] \in \Ch_\bbK$ is the cochain complex
that consists of 
\begin{subequations}\label{eqn:internalhom}
\begin{flalign}
[V, W]^n := \prod_{p \in \bbZ}{\Hom_\bbK(V^p, W^{n + p})} \quad,
\end{flalign}
for all $n \in \bbZ$, 
where $\Hom_\bbK$ denotes the vector space of linear maps, 
and of the differential	$\partial$ defined by
\begin{flalign}
\partial f := Q_W \circ f - (-1)^{\lvert f \rvert} \,f \circ Q_V \quad,
\end{flalign}
\end{subequations}
for all homogeneous $f \in [V, W]$. 
\sk

To each cochain complex $V \in \Ch_\bbK$ one can assign
its cohomology $H^\bullet (V) = (H^n (V))_{n \in \bbZ}$, that is the graded 
vector space defined degree-wise by
$H^n (V) := \Ker (Q_V^n) / \Imm (Q_V^{n - 1})$, for all $n \in \bbZ$. 
The compatibility of cochain maps with differentials entails that 
cohomology extends to a functor
$H^\bullet$ from $\Ch_\bbK$ to the category of graded vector spaces.
A cochain map $f : V \to W$ in $\Ch_\bbK$ is called 
a {\it quasi-isomorphism} if it induces an isomorphism 
$H^\bullet (f) : H^\bullet (V) \to H^\bullet (W)$ in cohomology.
In many circumstances quasi-isomorphic cochain complexes should be regarded as ``being the same'',
which can be made precise by using techniques from model category theory.
It is proven
in \cite{Hovey} that $\Ch_\bbK$ carries the structure of a
closed symmetric monoidal model category, whose weak equivalences are the quasi-isomorphisms and 
whose fibrations are the degree-wise surjective cochain maps. 
\begin{rem}
Let us briefly recall how one may interpret the cohomology of the internal hom 
$[V, W]\in\Ch_\bbK$ between cochain complexes $V, W \in \Ch_\bbK$
in terms of higher cochain homotopies. 
Given two $n$-cocycles $f, g \in \Ker(\partial^n)$ in $[V,W]$, one defines a
{\it cochain homotopy} from $f$ to $g$ as an $(n - 1)$-cochain $\lambda \in [V, W]^{n - 1}$ such that
$\partial \lambda = g - f$. 
Since $\partial \lambda \in \Imm(\partial^{n-1})$ is an $n$-coboundary in $[V, W]$, 
the cohomology classes $[f] = [g] \in H^n ([V, W])$ coincide
if and only if a cochain homotopy from $f$ to $g$ exists.
In particular, for $n = 0$ one recovers the ordinary concept of cochain
homotopies between two cochain maps $f, g : V \to W$ in $\Ch_\bbK$.
\end{rem}

Let us also fix our convention for shifts of cochain complexes. 
For a cochain complex $V \in \Ch_\bbK$ 
and an integer $q \in \bbZ$, we define the $q$-shift $V[q] \in \Ch_\bbK$ of $V$ 
as the cochain complex consisting of $V[q]^n := V^{q + n}$, for all $n \in \bbZ$, and of the differential
$Q_{V[q]} := (-1)^q \,Q_{V}$.
Note that $V[p][q] = V[p+q]$, for all $p, q \in \bbZ$, and that $V[0] = V$. 
Furthermore, recalling the definition of the tensor product~\eqref{eqn:tensorproduct}, one obtains natural
cochain isomorphisms $V[q] \cong \bbK[q] \otimes V$ for all $q \in \bbZ$.

\subsection{\label{subsec:pairing}Extension of pairings to symmetric algebras}
In this paper we will encounter various types of pairings $\tau \in [V\otimes V,\bbK]^p$
of degree $p\in\bbZ$ on a cochain complex $V\in\Ch_\bbK$.
These pairings are either symmetric or anti-symmetric,
i.e.\
\begin{flalign}
\tau\circ \gamma = s\,\tau
\end{flalign}
with $s=+1$ in the symmetric case and $s=-1$ in the anti-symmetric case,
where $\gamma$ denotes the symmetric braiding of $\Ch_\bbK$. 
In particular, we shall consider shifted and also unshifted 
(i.e.\ $0$-shifted) (linear) Poisson structures as defined below. 

\begin{defi}\label{defi:linPois}
A {\it $p$-shifted (linear) Poisson structure} on a cochain complex $V\in\Ch_\bbK$
consists of a symmetric (respectively, anti-symmetric) pairing $\tau \in [V \otimes V, \bbK]^p$ 
of odd (respectively, even) degree $p \in \bbZ$ that is closed $\partial \tau =0$
with respect to the internal hom differential \eqref{eqn:internalhom}.
\end{defi}

The aim of this subsection is to describe an extension
of such pairings to suitable bi-derivations and, in the symmetric case, 
to suitable Laplacians 
on the symmetric algebra $\Sym\,V\in\dgCAlg_\bbK$.
The latter is the commutative differential graded algebra
defined by $\Sym\,V =  \bigoplus_{n=0}^\infty \Sym^n\, V$,
with unit element $\oone :=1\in \Sym^0\,V=\bbK $ and multiplication
\begin{flalign}
\mu(v_1\cdots v_n\otimes v^\prime_1\cdots v^\prime_{n^\prime}):=
v_1 \cdots v_n\, v^\prime_1 \cdots v^\prime_{n^\prime}\quad,
\end{flalign}
for all $n, n^\prime \geq 0$ and all 
$v_1,\dots, v_n,v^\prime_1,\dots, v^\prime_{n^\prime} \in V$. 
(By convention, the length $n=0$ corresponds to the unit $\oone$.)
\begin{defi}\label{defi:endo}
Given an (anti-)symmetric pairing $\tau\in [V\otimes V,\bbK]^p$ 
of degree $p$, we define 
\begin{flalign}\label{eqn:endo}
\pair{-,-}_\tau \in \big[\Sym\, V \otimes \Sym\, V, \Sym\, V \otimes \Sym\, V\big]^p
\end{flalign}
as the unique graded linear map of degree $p$ that fulfills the following conditions: 
\begin{enumerate}[label=(\roman*)]
\item $\pair{-,-}_\tau$ is (anti-)symmetric, i.e.\
\begin{flalign}\label{eqn:bracket-sym}
\gamma \circ \pair{-,-}_\tau\circ \gamma = s \, \pair{-,-}_\tau 
\end{flalign}
with $s=+1$ in the symmetric case and $s=-1$ in the anti-symmetric case;

\item for all $v_1,v_2 \in V$, 
$\pair{v_1,v_2}_\tau = \tau(v_1 \otimes v_2)\, \oone \otimes \oone \in \Sym\, V \otimes \Sym\, V$; 

\item for all homogeneous $a \in \Sym\, V$, 
$\pair{a,-}_\tau : \Sym\, V \to \Sym\, V \otimes \Sym\, V$ 
is a graded derivation of degree $\lvert a \rvert + p$ 
with respect to the $(\Sym\, V)$-module structure on $\Sym\, V \otimes \Sym\, V$ 
given by multiplication on the second tensor factor, i.e.\
\begin{flalign}\label{eqn:bracket-der}
\pair{a, bc}_\tau = \pair{a,b}_\tau \, (\oone \otimes c) + (-1)^{(\lvert a \rvert + p) \lvert b \rvert} \,(\oone \otimes b)\,  \pair{a, c}_\tau\quad,
\end{flalign}
for all homogeneous $b,c \in \Sym\, V$.
\end{enumerate}	
\end{defi}

An immediate consequence of the previous definition is that 
\begin{flalign}\label{eqn:partial-endo}
\partial \pair{-,-}_\tau = \pair{-,-}_{\partial \tau} \quad. 
\end{flalign}
Furthermore, given two cochain complexes $V, W \in \Ch_\bbK$ endowed with
(anti-)symmetric pairings $\tau\in [V\otimes V,\bbK]^p$ and $\omega\in [W\otimes W,\bbK]^p$
of degree $p$ and a cochain map $f: V \to W$ in $\Ch_\bbK$ preserving them, i.e.\ $\tau = \omega \circ (f \otimes f)$, 
one has 
\begin{flalign}\label{eqn:f-endo}
(\Sym\, f \otimes \Sym\, f) \circ \pair{-,-}_\tau = \pair{-,-}_\omega \circ (\Sym\, f \otimes \Sym\, f) \quad. 
\end{flalign}

\begin{rem}\label{rem:Poiss-bracket}
A $p$-shifted (linear) Poisson structure $\tau$ on $V$ can be extended 
to a $p$-shifted Poisson bracket $\{-,-\}_\tau$ on $\Sym\, V$. Indeed, from 
Definitions~\ref{defi:linPois} and~\ref{defi:endo} it follows that 
\begin{flalign}
\{-,-\}_\tau := \mu \circ \pair{-,-}_\tau \,\in\, \big[ \Sym\, V \otimes \Sym\, V,  \Sym\, V\big]^p
\end{flalign}
defines a $p$-shifted Poisson bracket, i.e.\ a graded linear map of degree $p$ that is 
closed $\partial \{-,-\}_\tau = 0$, 
symmetric (respectively, anti-symmetric) for $p$ odd (respectively, even) 
and fulfills the graded Leibniz rule and the Jacobi identity. 
\end{rem}

\begin{defi}\label{defi:Laplacian}
Given a symmetric pairing $\tau\in[V\otimes V,\bbK]^p$ 
of degree $p$, we define the {\it Laplacian}
\begin{flalign}
\Delta_\tau \in \big[\Sym\, V, \Sym\, V\big]^p
\end{flalign}
as the unique graded linear map of degree $p$ that fulfills the following conditions:
\begin{enumerate}[label=(\roman*)]
\item $\Delta_\tau (\oone) = 0$;

\item for all $v \in V$, $\Delta_\tau (v) = 0$;

\item for all $v_1, v_2 \in V$, $\Delta_\tau (v_1\, v_2) = \tau(v_1 \otimes v_2)\, \oone$;

\item for all homogeneous $a, b \in \Sym\, V$, 
\begin{flalign}\label{eqn:Laplacian}
\Delta_\tau (a\, b) = \Delta_\tau (a)\, b + (-1)^{p \lvert a \rvert} \,a\, \Delta_\tau (b) +  \mu(\pair{a, b}_\tau)\quad.
\end{flalign}
\end{enumerate}
\end{defi}

The defining properties of $\Delta_\tau$ imply the explicit formula 
\begin{flalign}\label{eqn:Laplacian-explicit}
\Delta_\tau (v_1 \cdots v_n) = \sum_{i < j}{(-1)^{p \sum_{k=1}^{i - 1}{\lvert v_k \rvert} + \lvert v_j \rvert \sum_{k= i +1}^{j - 1}{\lvert v_k \rvert}}\, \tau(v_i \otimes v_j)\, v_1 \cdots \check{v}_i \cdots \check{v}_j \cdots v_n } \quad, 
\end{flalign}
for all $n \geq 1$ and all homogeneous $v_1, \ldots, v_n \in V$, 
where $\check{\cdot}$ means to omit the corresponding factor. 
Furthermore, for $p$ even, iterating \eqref{eqn:Laplacian} and observing 
that both $\Delta_\tau \otimes \id$ and $\id \otimes \Delta_\tau$ 
graded commute with $\pair{-,-}_\tau$, one finds that, for all $n \geq 1$, 
\begin{flalign}\label{eqn:Deltan-mu}
\Delta_\tau^n\circ \mu = \mu\circ \big(\Delta_{\tau\, \otimes} + \pair{-,-}_{\tau}\big)^n = 
\sum_{k=0}^n \binom{n}{k}\, \mu\circ \pair{-,-}_\tau^{n-k}\circ \Delta_{\tau\, \otimes}^k \quad,
\end{flalign}
where $\Delta_{\tau\, \otimes} := \Delta_\tau \otimes \id + \id \otimes \Delta_\tau$. 
(For $n=1$ this recovers \eqref{eqn:Laplacian}. 
For $p$ odd, the left-hand side vanishes identically for $n \geq 2$,
see \eqref{eqn:LaplacianLaplacian} below.) 
Taking also \eqref{eqn:partial-endo} into account, one shows that 
\begin{flalign}\label{eqn:partial-Laplacian}
\partial \Delta_\tau = \Delta_{\partial \tau} \quad.
\end{flalign}
Given two symmetric pairings $\tau\in[V\otimes V,\bbK]^p$ 
and $\tau^\prime\in[V\otimes V,\bbK]^{p^\prime}$ 
of degrees $p$ and $p^\prime$ respectively, 
the explicit formula~\eqref{eqn:Laplacian-explicit} for the Laplacian 
entails
\begin{flalign}\label{eqn:LaplacianLaplacian}
\Delta_{\tau}\circ \Delta_{\tau^\prime} = (-1)^{p p^\prime} \,\Delta_{\tau^\prime}\circ \Delta_{\tau}\quad.
\end{flalign}
Furthermore, given two cochain complexes $V, W \in \Ch_\bbK$ endowed with 
symmetric pairings $\tau\in[ V \otimes V,\bbK]^p$ 
and $\omega\in [ W \otimes W,\bbK]^p$ of degree $p$  
and a cochain map $f: V \to W$ in $\Ch_\bbK$ preserving them, i.e.\ $\tau = \omega \circ (f \otimes f)$, 
it follows from \eqref{eqn:Laplacian-explicit} that 
\begin{flalign}\label{eqn:f-Laplacian}
\Sym\, f \circ \Delta_\tau = \Delta_\omega \circ \Sym\, f \quad. 
\end{flalign}

\subsection{\label{subsec:Green}Lorentzian geometry and Green's operators}
In this subsection we recall some relevant concepts from Lorentzian 
geometry and Green hyperbolic differential operators. 
We refer to \cite{Bar, BGP, ONeill} for an in-depth introduction to these topics.
\sk

A {\it Lorentzian manifold} $(M,g)$ is a smooth manifold $M$ 
endowed with a metric $g$ of signature $(-,+,\dots,+)$.
Given a non-zero tangent vector $0 \neq v \in T_x M$ at a point 
$x \in M$, we say that $v$ is 
{\it spacelike} if $g(v, v) > 0$, {\it lightlike} if $g(v, v) = 0$ 
and {\it timelike} if $g(v, v) < 0$. 
$v$ is also called {\it causal} if $g(v, v) \leq 0$, 
that is $v$ is either timelike or lightlike.
Let $ I \subseteq \bbR$ be an open interval. A curve $c : I \to M$ is 
called {\it spacelike} ({\it lightlike}, {\it timelike} or {\it causal}) 
if its tangent vectors $\dot{c}(t)$ are spacelike (lightlike, timelike 
or causal, respectively), for all $t \in I$.
A Lorentzian manifold $M$ is called {\it time-orientable} if there exists 
an everywhere timelike vector field $\ttt \in \Gamma (TM)$. 
Such $\ttt$ determines a {\it time-orientation} on $M$. 
We will denote {\it oriented and time-oriented Lorentzian manifolds} 
by $M = (M,g,\mathfrak{o},\ttt)$, where $\mathfrak{o}$ is the chosen orientation. 
A timelike or causal curve $c : I \to M$ is said to be 
{\it future directed} if $g(\ttt, \dot{c}) < 0$ 
and {\it past directed} if $g(\ttt, \dot{c}) > 0$.
The {\it chronological future/past} $I_M^\pm (S) \subseteq M$ 
of a subset $S \subseteq M$ consists of all points
that can be reached by a future/past directed timelike curve stemming from $S$.
Similarly, the {\it causal future/past} $J_M^\pm (S) \subseteq M$ 
consists of $S$ itself and of all points that can be reached 
by a future/past directed causal curve stemming from $S$. 
By definition, $I_M^\pm(S) \subseteq J_M^\pm(S)$; moreover, recall
from e.g.\ \cite[Chapter 14]{ONeill} that 
\begin{flalign}
I_M^\pm(J_M^\pm(S)) = I_M^\pm(S) = J_M^\pm(I_M^\pm(S)) \subseteq M 
\end{flalign}
is always an open subset. 
A subset $S \subseteq M$ is called {\it causally convex}
if $J_M^+ (S) \cap J_M^- (S) \subseteq S$, i.e.\ when 
all causal curves with endpoints in $S$ lie in $S$. 
An example of a causally convex subset is the 
{\it causally convex hull} 
\begin{flalign}
J_M^{+\cap-}(S) := J_M^+(S) \cap J_M^-(S) \subseteq M
\end{flalign}
of a subset 
$S \subseteq M$, i.e.\ the smallest causally convex subset of $M$ 
that contains $S$. 
\begin{defi}\label{defi:globhyp}
An oriented and time-oriented Lorentzian manifold $M$ is called 
{\it globally hyperbolic} if it admits a {\it Cauchy surface} 
$\Sigma \subset M$, i.e.\ a subset that is met exactly once 
by any inextendible future directed timelike curve in $M$. 
$\Loc_m$ denotes the category whose objects are all 
$m$-dimensional oriented and time-oriented globally hyperbolic Lorentzian manifolds $M$ and whose
morphisms are all orientation and time-orientation preserving 
isometric embeddings $f : M \to M^\prime$
with open and causally convex image $f (M) \subseteq M^\prime$.
\end{defi}

For $M \in \Loc_m$ and $O \subseteq M$ open, one has that the 
causal future/past 
\begin{flalign}
J_M^\pm(O) = I_M^\pm(O) 
\end{flalign}
coincides with the chronological one. 
(Indeed, any $p \in J_M^\pm(O)$ lies along a future/past 
directed causal curve emanating from some $q \in O$. Since $O$ is open, 
$q$ can be reached via a future/past directed timelike curve 
emanating from some $r \in O$. But then 
$p \in J_M^{\pm}(q) \subseteq J_M^{\pm}(I_M^\pm(r)) = I_M^\pm(r) \subseteq I_M^\pm(O)$.) 
In particular, when $O \subseteq M$ is open, 
the causal future/past $J_M^\pm(O) \subseteq M$ and the causally 
convex hull $J_M^{+\cap-}(O) \subseteq M$ are open subsets. 
\sk

Consider an oriented and time-oriented globally hyperbolic Lorentzian manifold $M \in \Loc_m$ of 
dimension $m \geq 2$. Let $E \to M$ be a real or complex vector bundle 
of finite rank. Denote the vector space of smooth sections of $E$ 
by $\Gamma(E)$ and the vector subspace of compactly supported sections 
by $\Gamma_\cc(E) \subseteq \Gamma(E)$.
\begin{defi}\label{defi:Gpm}
A {\it Green hyperbolic operator} is a linear differential operator 
$P : \Gamma (E) \to \Gamma (E)$ that admits
{\it retarded and advanced Green's operators} $G_\pm$, 
which are linear maps $G_\pm : \Gamma_\cc (E) \to \Gamma (E)$ such that, 
for all $\varphi \in \Gamma_\cc (E)$, the following conditions hold:
\begin{enumerate}[label={(\roman*)}]
\item $P G_\pm \varphi = \varphi$;

\item $G_\pm P \varphi = \varphi$;

\item $\supp (G_\pm \varphi) \subseteq J_M^\pm (\supp (\varphi))$.
\end{enumerate}
The difference $G := G_+ - G_- : \Gamma_\cc (E) \to \Gamma (E)$
between the retarded and advanced Green's operators is called the 
{\it retarded-minus-advanced propagator} and 
their average $G_D := \tfrac{1}{2}(G_+ + G_-)  : \Gamma_\cc (E) \to \Gamma (E)$
is called the {\it Dirac propagator}.
\end{defi}

In \cite{Bar} it is shown that the retarded and advanced Green's operators 
associated with a Green hyperbolic operator are unique.
\sk

Given a real vector bundle $E \to M$ endowed with a {\it fiber metric} 
$\langle -, - \rangle$, i.e.\ a fiber-wise non-degenerate, symmetric, 
bilinear form, and denoting the volume form on $M$ by $\vol_M$, 
one defines the integration pairing
\begin{flalign}\label{eqn:fibintpairing}
\langle\!\langle \varphi, \varphi^\prime \rangle\!\rangle := \int_M{\langle \varphi, \varphi^\prime \rangle\, \vol_M} \quad,
\end{flalign}
for all sections $\varphi, \varphi^\prime \in \Gamma (E)$ with compact 
overlapping support, i.e.\ such that 
$\supp(\varphi) \cap \supp(\varphi^\prime)\subseteq M$ is compact. 
Given two vector bundles $(E_1, \langle -, - \rangle_1), (E_2, \langle -, - \rangle_2)$  
endowed with fiber metrics and 
a linear differential operator $Q : \Gamma (E_1) \to \Gamma (E_2)$,
one defines its {\it formal adjoint} 
$Q^\ast : \Gamma (E_2) \to \Gamma (E_1)$ as the unique linear
differential operator such that 
\begin{flalign}
\langle\!\langle Q^\ast \varphi_2, \varphi_1  \rangle\!\rangle_1 := \langle\!\langle \varphi_2, Q \varphi_1 \rangle\!\rangle_2 \quad,
\end{flalign}
for all sections $\varphi_1 \in \Gamma (E_1)$, $\varphi_2 \in \Gamma (E_2)$ 
with compact overlapping support. A linear differential operator 
$P : \Gamma (E) \to \Gamma (E)$ on $(E, \langle -, - \rangle)$ 
is {\it formally self-adjoint} if $P^\ast = P$. 
When $P : \Gamma (E) \to \Gamma (E)$ is a formally self-adjoint Green hyperbolic operator,
the associated retarded and advanced Green's operators $G_\pm$ 
are ``formal adjoints'' of each other, i.e.\
\begin{flalign}\label{eqn:Gpmadjoint}
\langle\!\langle G_\pm \varphi, \varphi^\prime \rangle\!\rangle = \langle\!\langle \varphi, G_\mp \varphi^\prime \rangle\!\rangle \quad,
\end{flalign}
for all compactly supported sections 
$\varphi, \varphi^\prime \in \Gamma_\cc (E)$. 
This entails that the retarded-minus-advanced propagator 
$G$ is ``formally skew-adjoint'', i.e.\
\begin{flalign}\label{eqn:Gskewadjoint}
\langle\!\langle G \varphi, \varphi^\prime \rangle\!\rangle = - \langle\!\langle \varphi, G \varphi^\prime \rangle\!\rangle \quad ,
\end{flalign}
for all compactly supported sections 
$\varphi, \varphi^\prime \in \Gamma_\cc (E)$.

\subsection{\label{subsec:AQFTandtPFA}Algebraic QFTs and time-orderable prefactorization algebras}
Algebraic quantum field theories (AQFTs) \cite{BFV, FV, operads_aqft}
and factorization algebras \cite{CG, CG2, FA-vs-AQFT} provide two 
axiomatic frameworks to describe the algebraic structures 
on the observables of a quantum field theory in various geometric settings.
In this subsection we review some basic concepts from these two frameworks 
in the Lorentzian setting. 
\sk

We say that two $\Loc_m$-morphisms $f_1 : M_1 \to N \gets M_2 : f_2$ 
to a common target are {\it causally disjoint} if 
there exists no causal curve in $N$ connecting their images, i.e.\ 
$J_N (f_1(M_1)) \cap f_2(M_2) = \emptyset$, 
where $J_M(S) := J^+_M(S) \cup J^-_M(S)$ denotes the union 
of the causal future and past of a subset $S \subseteq M$. 
Furthermore, a morphism $f : M \to N$ in $\Loc_m$ is {\it Cauchy}
if its image $f(M) \subseteq N$ contains a Cauchy surface of $N$.
\begin{defi}\label{defi:AQFT}
A $\Ch_\bbC$-valued {\it algebraic quantum field theory (AQFT)} $\A$ on $\Loc_m$ is a functor 
$\A : \Loc_m \to \dgAlg_\bbC$ taking values in the category $\dgAlg_\bbC$ 
of differential graded algebras that satisfies the following axioms:
\begin{enumerate}[label=(\roman*)]
\item {\it Einstein causality:} For all causally disjoint morphisms 
$f_1 : M_1 \to N \gets M_2 : f_2$ in $\Loc_m$, 
the diagram
\begin{flalign}\label{eqn:Einstein_causality}
\xymatrix{
\A (M_1) \otimes \A (M_2) \ar[rr]^-{\A (f_1) \otimes \A (f_2)} \ar[d]_-{\A (f_1) \otimes \A (f_2)} && \A (N) \otimes \A (N) \ar[d]^-{\mu_N^\op} \\
\A (N) \otimes \A (N) \ar[rr]_-{\mu_N} && \A (N)
}
\end{flalign}
in $\Ch_\bbC$ commutes, where $\mu_N$ and $\mu_N^\op := \mu_N \circ \gamma$ 
are the multiplication and the opposite multiplication of $\A (N) \in \dgAlg_\bbC$;

\item {\it Time-slice:} For all Cauchy morphisms $f : M \to N$ in $\Loc_m$, 
the morphism $\A(f) : \A(M) \to \A(N)$ in $\dgAlg_\bbC$ is a quasi-isomorphism.
\end{enumerate}
A morphism $\kappa : \A \to \B$ between AQFTs
is a natural transformation. 
This defines the category $\AQFT$ of $\Ch_\bbC$-valued AQFTs
as the full subcategory $\AQFT \subseteq \dgAlg_\bbC^{\Loc_m}$ of the functor category 
consisting of all functors that satisfy the Einstein causality and time-slice axioms.
\end{defi}
\begin{rem}\label{rem:hAQFT}
There exists a more elegant and powerful operadic
description \cite{operads_aqft} of the category $\AQFT$.
This more abstract perspective is particularly useful to endow $\AQFT$ 
with a model category structure \cite{hAQFT}, which 
provides a solid foundation for the study of $\Ch_\bbC$-valued
AQFTs. To prove the 
results of our present paper, we do not have to make
explicit use of these techniques.
\end{rem}

Our goal is to construct and compare AQFTs
and prefactorization algebras in the Lorentzian setting. 
For this purpose, we recall below 
a Lorentzian version of the prefactorization algebras from \cite{CG},
called {\it time-orderable} prefactorization algebras \cite{FA-vs-AQFT}. 
This requires some preliminaries.
A tuple of $\Loc_m$-morphisms $(f_1 : M_1 \to N, \dots, f_n : M_n \to N)$, 
also denoted $\und{f}: \und{M} \to N$, to a common target is called 
{\it time-ordered} if $J_N^+ (f_i(M_i)) \cap f_j(M_j) = \emptyset$, for all $i < j$. 
Given a tuple $\und{f}: \und{M} \to N$ in $\Loc_m$ of length $n$, 
a {\it time-ordering permutation} $\rho \in \Sigma_n$ 
is a permutation such that the $\rho$-permuted tuple 
$\und{f} \rho := (f_{\rho (1)}, \dots, f_{\rho (n)}) : \und{M} \rho \to N$ 
of $\Loc_m$-morphisms is time-ordered. 
When a time-ordering permutation exists, one says that 
$\und{f}: \und{M} \to N$ in $\Loc_m$ is {\it time-orderable}. 
(Note that the time-ordering permutation for a tuple may not be unique. For instance, 
two morphisms $f_1: M_1 \to N \leftarrow M_2: f_2$ in $\Loc_m$ 
are causally disjoint precisely when 
both $(f_1,f_2)$ and $(f_2,f_1)$ are time-ordered pairs.)
A time-orderable 1-tuple $(f) : \und{M} \to N$ in $\Loc_m$ 
is denoted simply as a morphism $f : M \to N$ in $\Loc_m$ and, for each $N \in \Loc_m$, 
we define a unique time-orderable empty tuple $\emptyset \to N$. 
Time-orderable tuples are composable and carry permutation 
group actions, see \cite{FA-vs-AQFT}. These facts are crucial 
for the next definition. 
\begin{defi}\label{defi:tPFA}
A $\Ch_\bbC$-valued {\it time-orderable prefactorization algebra} $\F$ 
on $\Loc_m$ consists of the data listed below: 
\begin{enumerate}[label=(\alph*)]
\item For each $M \in \Loc_m$, a cochain complex $\F (M) \in \Ch_\bbC$.

\item For each time-orderable tuple $\und{f} : \und{M} \to N$ in $\Loc_m$, 
a morphism $\F (\und{f}) : \F(\und{M}) \to \F (N)$ 
in $\Ch_\bbC$, called {\it time-ordered product}, 
where $\F(\und{M}) := \bigotimes_{i = 1}^n{\F (M_i)} \in \Ch_\bbC$ 
denotes the tensor product.
By convention, the time-ordered product assigned to an empty tuple $\emptyset \to N$ is a morphism 
$\bbC \to \F (N)$ in $\Ch_\bbC$ from the monoidal unit.
\end{enumerate}
These data are subject to the following axioms:
\begin{enumerate}[label=(\roman*)]
\item  For all time-orderable tuples $\und{f} = (f_1, \dots, f_n) : \und{M} \to N$ and 
$\und{g}_i = (g_{i1}, \dots, g_{ik_i}): \und{L}_i \to M_i$, $i = 1, \dots, n$, in $\Loc_m$, the diagram
\begin{flalign}\label{eqn:composition-factor}
\xymatrix@C=4em{
\bigotimes\limits_{i = 1}^n{\F(\und{L}_i)} \ar[r]^-{\bigotimes_i{\F (\und{g}_i)}} \ar[dr]_(0.45)*+<1em>{\labelstyle \F (\und{f}(\und{g}_1, \dots, \und{g}_n))} 
& \F (\und{M}) \ar[d]^-{\F (\und{f})} \\
&  \F (N)
}
\end{flalign}
in $\Ch_\bbC$ commutes, where 
$\und{f}(\und{g}_1, \dots, \und{g}_n) := (f_1 g_{11}, \dots, f_n g_{nk_n}) : (\und{L}_1, \dots, \und{L}_n) \to N$
is the time-orderable tuple given by composition in $\Loc_m$.

\item For all $M \in \Loc_m$, $\F (\id_M) = \id_{\F (M)} : \F (M) \to \F (M)$ in $\Ch_\bbC$ is the identity.

\item For all time-orderable tuples $\und{f} : \und{M} \to N$ in $\Loc_m$
and permutations $\sigma \in \Sigma_n$, the diagram
\begin{flalign}\label{eqn:permutation-factor}
\xymatrix{
\F (\und{M}) \ar[d]_{\gamma_\sigma} \ar[rr]^-{\F (\und{f})} && \F (N) \\
\F (\und{M}\sigma) \ar[urr]_-*+<0.4em>{\labelstyle \F (\und{f}\sigma)}
}
\end{flalign}
in $\Ch_\bbC$ commutes, where $\gamma_\sigma$ is defined by 
the symmetric braiding $\gamma$ of $\Ch_\bbC$.
\end{enumerate} 
We say that a time-orderable prefactorization algebra $\F$ satisfies the {\it time-slice axiom} if, 
for all Cauchy morphism $f : M \to N$ in $\Loc_m$, 
$\F(f) : \F(M) \to \F(N)$ in $\Ch_\bbC$ is a quasi-isomorphism. 
\sk

A morphism $\zeta = (\zeta_M)_{M \in \Loc_m} : \F \to \G$ 
of time-orderable prefactorization algebras is a collection of cochain maps 
$\zeta_M : \F (M) \to \G (M)$ in $\Ch_\bbC$, indexed by objects $M \in \Loc_m$,
that is compatible with the time-ordered products in the sense that, 
for all time-orderable tuples $\und{f} : \und{M} \to N$ in $\Loc_m$, the diagram
\begin{flalign}
\xymatrix{
\F (\und{M}) \ar[rr]^-{\F (\und{f})} \ar[d]_-{{\zeta_{\und{M}}}}&& \F (N) \ar[d]^-{\zeta_N} \\
\G (\und{M}) \ar[rr]_-{\G (\und{f})} && \G (N)
}
\end{flalign}
in $\Ch_\bbC$ commutes, where $\zeta_{\und{M}} := \bigotimes_i \zeta_{M_i}$. 
We denote the category of time-orderable prefactorization algebras 
on $\Loc_m$ satisfying the time-slice axiom by $\tPFA$. 
\end{defi}


\section{\label{sec:Witness}Green's witnesses}
In this section we briefly recall the concept of a {\it Green's witness} for a complex of linear differential operators,
see \cite{GreenHyp} for more details. This consists of a collection of degree decreasing linear differential operators
that enable the explicit construction of 
{\it retarded and advanced Green's homotopies}. 
The latter are differential graded analogs of the usual retarded and advanced Green's operators, 
see e.g.\ \cite{BGP,Bar} and also Subsection \ref{subsec:Green}, 
and they will play a key role in our
construction of AQFTs and
their comparison to time-orderable prefactorization algebras in Section~\ref{sec:Quant}.
Given a Green's witness, we shall 
endow the underlying complex of linear differential operators with the following three structures: 
1.)~a $(-1)$-shifted Poisson structure $\tau_{(-1)}$, 2.)~an unshifted
Poisson structure $\tau_{(0)}$ and 3.)~a symmetric pairing 
$\tau_D$,  that we call Dirac pairing, trivializing the $(-1)$-shifted Poisson structure, i.e.\ $\tau_{(-1)} = \partial \tau_D$.
We shall show that $\tau_{(-1)}$, $\tau_{(0)}$ and $\tau_D$ are natural 
when all input data are natural (with respect to 
the category $\Loc_m$ of $m$-dimensional oriented and time-oriented 
globally hyperbolic Lorentzian manifolds). 
In particular, we shall construct a functor 
$(\FFF_\cc[1],\tau_{(0)}) : \Loc_m \to \PCh_\bbR$ that assigns to each 
$M \in \Loc_m$ a Poisson cochain complex $(\FFF_\cc(M)[1],\tau^M_{(0)})$
whose cochains may be interpreted field-theoretically as linear observables.
(Here $\PCh_\bbR$ denotes the category whose objects are Poisson cochain 
complexes $(V,\tau)$, consisting of a cochain complex $V \in \Ch_\bbR$ endowed 
with an unshifted linear Poisson structure $\tau$, see Definition~\ref{defi:linPois}, 
and whose morphisms $f: (V,\tau) \to (W,\omega)$ are cochain maps 
$f: V \to W$ in $\Ch_\bbR$ preserving the Poisson structures, 
i.e.\ $\omega \circ (f \otimes f) = \tau$.) 
In Theorem~\ref{th:time-slice-cl} we shall prove that the functor $(\FFF_\cc[1],\tau_{(0)})$ satisfies the classical 
analogs of the Einstein causality and time-slice axioms.

\subsection{\label{subsec:taus}\texorpdfstring{$\tau_{(-1)}$, $\tau_{(0)}$ and $\tau_D$}{Structures} over a fixed globally hyperbolic Lorentzian manifold}
Given a ($\bbZ$-)graded ($\bbR$-)vector bundle $F \to M$ (degree-wise of finite rank)
over an oriented and time-oriented globally hyperbolic Lorentzian manifold $M \in \Loc_m$,
we denote by
\begin{flalign}
\FFF(M)^n := \Gamma(F^n)
\end{flalign}
the vector space of degree $n$ smooth sections, 
i.e.\ the smooth sections of the degree $n$ vector bundle $F^n \to M$, and by
\begin{flalign}
\FFF_\cc(M)^n := \Gamma_\cc (F^n)
\end{flalign}
the vector space of degree $n$ smooth sections with compact support.
\begin{defi}\label{defi:CLDO-fibermetric}
A {\it complex of linear differential operators} $(F,Q)$ over $M \in \Loc_m$ consists of
a graded vector bundle $F \to M$ and of a collection 
$Q = (Q^n : \FFF(M)^n \to \FFF(M)^{n + 1})_{n \in \bbZ}$
of degree increasing linear differential operators such that $Q^{n + 1} Q^n = 0$, for all $n \in \bbZ$. 
We denote by $\FFF(M) \in \Ch_\bbR$ the cochain complex of sections associated
with the complex of linear differential operators $(F,Q)$.
\sk
	
A {\it compatible $(-1)$-shifted fiber metric} $(-,-)$ on $(F,Q)$ is
a fiber-wise non-degenerate, graded anti-symmetric, graded vector bundle map
$(-,-) : F \otimes F \to M \times \bbR[-1]$ such that the identity
\begin{flalign}\label{eqn:compatibility-pairing}
\int_M{(Q\varphi_1, \varphi_2)\,\vol_M} + (-1)^{\lvert \varphi_1 \rvert} \int_M{(\varphi_1, Q\varphi_2)\, \vol_M} = 0
\end{flalign}
holds for all homogeneous sections $\varphi_1, \varphi_2 \in \FFF(M)$ 
with compact overlapping support.
\end{defi}

\begin{rem}
The compatibility condition~\eqref{eqn:compatibility-pairing} implies 
that the integration pairing
\begin{subequations}\label{eqn:intpairing}
\begin{flalign}
(\!(-,-)\!) : \FFF_\cc(M) \otimes \FFF(M) \longrightarrow \bbR[-1] \quad,
\end{flalign}
defined by
\begin{flalign}
(\!(\psi, \varphi)\!) := \int_M{(\psi, \varphi)\,\vol_M} \quad,
\end{flalign}
\end{subequations}
for all $\psi \in \FFF_\cc(M)$ and $\varphi \in \FFF(M)$, is a cochain map.
\end{rem}

\begin{defi}\label{defi:witness}
A {\it (formally self-adjoint) Green's witness} $W = (W^n)_{n \in \bbZ}$ 
for a complex of linear differential operators $(F,Q)$ endowed with a 
compatible $(-1)$-shifted fiber metric $(-,-)$ consists of a collection 
of degree decreasing linear differential operators 
$W^n : \FFF(M)^n \to \FFF(M)^{n-1}$ such that the following conditions
hold:
\begin{enumerate}[label={(\roman*)}]
\item For all $n \in \bbZ$, 
$P^n := Q^{n-1}\, W^n + W^{n+1}\, Q^n : \FFF(M)^n \to \FFF(M)^n$
are Green hyperbolic operators.

\item $Q\, W\, W = W\, W\, Q$.

\item $\int_M{(W \varphi_1, \varphi_2)\,\vol_M} = (-1)^{\lvert \varphi_1 \rvert} \int_M{(\varphi_1, W\varphi_2)\,\vol_M}$, 
for all homogeneous sections $\varphi_1, \varphi_2 \in \FFF(M)$ 
with compact overlapping support.
\end{enumerate}
\end{defi}

\begin{rem}\label{rem:propW}
Some direct consequences of Definition~\ref{defi:witness} are listed below:
\begin{enumerate}[label={(\arabic*)}]
\item For all $n \in \bbZ$, there exist unique retarded and advanced Green's operators 
$G^n_\pm : \FFF_\cc(M)^n \to \FFF(M)^n$ associated with the Green hyperbolic operators $P^n$;

\item It follows that $P\, W = W\, P$ and $P\, Q = Q\, P$, 
hence also $G_\pm\, W = W\, G_\pm$ and $G_\pm\, Q = Q\, G_\pm$;

\item $P$ is formally self-adjoint, i.e.\
$\int_M{(P \varphi_1, \varphi_2)\,\vol_M} = \int_M{(\varphi_1, P \varphi_2)\,\vol_M}$, 
for all sections $\varphi_1, \varphi_2 \in \FFF(M)$ 
with compact overlapping support. It follows that 
$\int_M{(\psi_1, G_\pm \psi_2)\,\vol_M} = \int_M{(G_\mp \psi_1, \psi_2)\,\vol_M}$, 
for all sections $\psi_1, \psi_2 \in \FFF_\cc(M)$
with compact support, and hence also 
$\int_M{(\psi_1, G \psi_2)\,\vol_M} = - \int_M{(G \psi_1, \psi_2)\,\vol_M}$ and 
$\int_M{(\psi_1, G_D \psi_2)\,\vol_M} = \int_M{(G_D \psi_1, \psi_2)\,\vol_M}$,
where $G := G_+ - G_-$ and $G_D := \tfrac{1}{2}(G_+ + G_-)$ 
denote respectively the retarded-minus-advanced and Dirac propagators.
\end{enumerate}
These observations will be used frequently in our constructions in this paper.
\end{rem}

In analogy with the Riemannian setting \cite{CG2}, 
we introduce the following terminology.
\begin{defi}\label{defi:BV}
A {\it free BV theory} $(F,Q,(-,-),W)$ on $M \in \Loc_m$ consists of 
a complex of linear differential operators $(F,Q)$ with a 
compatible $(-1)$-shifted fiber metric $(-,-)$ and a Green's witness $W$. 
\end{defi}

Several examples of free BV theories, closely related to the examples 
from \cite{linYM,hrep,GreenHyp}, are presented below. 
\begin{ex}\label{ex:P}
Our first example of a free BV theory over $M \in \Loc_m$ is obtained 
from an ordinary field theory, which is defined by a formally 
self-adjoint Green hyperbolic operator $P$ 
acting on sections of a vector bundle $E$ over $M$ 
endowed with a fiber metric $\langle-,-\rangle$. 
To these data one assigns the free BV theory $(F_P,Q_P,(-,-)_P,W_P)$ 
consisting of the complex of linear differential operators 
\begin{subequations}
\begin{flalign}
(F_P,Q_P) := \big( \xymatrix{
E \ar[r]^-{P} & E
} \big) 
\end{flalign}
concentrated in degrees $0$ and $1$, of the compatible $(-1)$-shifted 
fiber metric $(-,-)_P$ uniquely determined by 
\begin{flalign}
(\varphi^\ddagger,\varphi)_P := \langle \varphi^\ddagger,\varphi \rangle \quad, 
\end{flalign}
for all $\varphi \in F_P^0 = E$ and $\varphi^\ddagger \in F_P^1 = E$ 
over the same base point, and of the 
Green's witness 
\begin{flalign}
W_P := \big( \xymatrix{E & E \ar[l]_-{\id}} \big) \quad. 
\end{flalign}
\end{subequations}
Here and in the following examples we decided to use a convenient graphical visualization
for a Green's witness as a sequence of linear differential operators, 
which is pointing from right to left because $W$ decreases the degree. It is important to emphasize that 
this sequence is in general {\em not} a chain complex because 
a Green's witness is not necessarily square-zero.
\end{ex}

\begin{ex}\label{ex:CS}
The free BV theory $(F_\CS,Q_\CS,(-,-)_\CS,W_\CS)$ 
associated with linear Chern-Simons theory on $M \in \Loc_3$ 
consists of the complex of linear differential operators 
\begin{subequations}
\begin{flalign}
(F_\CS,Q_\CS) := \big( \xymatrix{
\Lambda^0 M \ar[r]^-{\dd} & \Lambda^1 M \ar[r]^-{\dd} & \Lambda^2 M \ar[r]^-{\dd} & \Lambda^3 M
} \big)
\end{flalign}
concentrated between degrees $-1$ and $2$ (this is the $1$-shift 
of the de Rham complex up to a global sign), 
of the $(-1)$-shifted fiber metric $(-,-)_\CS$ uniquely determined by 
\begin{flalign}
(A^\ddagger,A)_\CS := \ast^{-1}(A^\ddagger \wedge A) \quad, \qquad (c^\ddagger,c)_\CS := - \ast^{-1}(c^\ddagger \wedge c) \quad, 
\end{flalign}
for all $c \in F_\CS^{-1} = \Lambda^0 M$, 
$A \in F_\CS^0 = \Lambda^1 M$, 
$A^\ddagger \in F_\CS^1 = \Lambda^2 M$ 
and $c^\ddagger \in F_\CS^2 = \Lambda^3 M$ over the same base point, 
where $\wedge$ denotes the wedge product on differential forms 
and $\ast$ denotes the Hodge operator on $M$, and of the Green's witness 
\begin{flalign}
W_\CS := \big( \xymatrix{\Lambda^0 M & \Lambda^1 M \ar[l]_-{\delta} & \Lambda^2 M \ar[l]_-{\delta} & \Lambda^3 M \ar[l]_-{\delta}} \big) \quad, 
\end{flalign}
where $\delta := (-1)^k \ast^{-1} \dd\, \ast$ 
denotes the de Rham codifferential on $M$ on $k$-forms, for $k=1,2,3$. 
(It is useful to keep in mind that $\ast^{-1} = - \ast$ 
in odd dimension and Lorentzian signature.)
\end{subequations}
\end{ex}

\begin{ex}\label{ex:MW}
The free BV theory $(F_\MW,Q_\MW,(-,-)_\MW,W_\MW)$ 
associated with Maxwell $p$-forms on $M \in \Loc_m$, 
for $p=0,\ldots,m-1$, consists of the complex of linear differential operators 
\begin{subequations}
\begin{flalign}
(F_\MW,Q_\MW) := \big( \xymatrix{\Lambda^0 M \ar[r]^-{\dd} & \cdots \ar[r]^-{\dd} & \Lambda^p M \ar[r]^-{\delta \dd} & \Lambda^p M \ar[r]^-{\delta} & \cdots \ar[r]^-{\delta} & \Lambda^0 M} \big)
\end{flalign}
concentrated between degrees $-p$ and $p+1$, 
of the $(-1)$-shifted fiber metric $(-,-)_\MW$ uniquely determined by 
\begin{flalign}
(a^\ddagger,a)_\MW := s_{k+1} \ast^{-1}(a^\ddagger \wedge \ast a) \quad,
\end{flalign}
for all $k=0,\ldots,p$, $a \in F_\MW^{-k} = \Lambda^{p-k} M$ and 
$a^\ddagger \in F_\MW^{k+1} = \Lambda^{p-k} M$ over the same base point, 
where $s_1 := 1$ and $s_k := (-1)^k s_{k-1}$, for $k=2,\ldots,p+1$,
and of the Green's witness 
\begin{flalign}
W_\MW := \big( \xymatrix{\Lambda^0 M & \cdots \ar[l]_-{\delta} & \Lambda^p M \ar[l]_-{\delta} & \Lambda^p M \ar[l]_-{\id} & \cdots \ar[l]_-{\dd} & \Lambda^0 M \ar[l]_-{\dd}} \big) \quad. 
\end{flalign}
Note that for $p=1$ Maxwell $p$-forms recover linear Yang-Mills theory. 
\end{subequations}
\end{ex}

Let $(F,Q,(-,-),W)$ be free BV theory. 
We define the {\it retarded/advanced Green's homotopy}
\begin{flalign}\label{eqn:Lambdapm}
\Lambda_\pm := W\, G_\pm = G_\pm\, W \in [\FFF_\cc(M), \FFF(M)]^{-1}\quad,
\end{flalign}
where $G_\pm$ denotes the retarded/advanced Green's operator 
associated with $P$, see Definition~\ref{defi:witness} 
and Remark~\ref{rem:propW}. (In \eqref{eqn:Lambdapm}
we used Remark~\ref{rem:propW}~(2) and that $W$ preserves supports.) 
Note that the retarded/advanced Green's homotopy $\Lambda_\pm \in [\FFF_\cc(M), \FFF(M)]^{-1}$
is a cochain homotopy that trivializes the cochain map
$j : \FFF_\cc(M) \to \FFF(M)$ in $\Ch_\bbR$ forgetting compact supports. 
More explicitly, one computes 
\begin{flalign}\label{eqn:delLambda}
\partial \Lambda_\pm =  Q\, W\, G_\pm + W\, G_\pm\, Q = P\, G_\pm = j \quad,
\end{flalign}
where in the first step we used the definition of the internal hom 
differential $\partial$, the second step follows from 
Remark~\ref{rem:propW}~(2) and in the last step we used that 
$G_\pm$ is the retarded/advanced Green's operator associated with $P$.

\begin{rem}\label{rem:Green-homotopies}
$\Lambda_\pm$ as defined in \eqref{eqn:Lambdapm} is a specific choice 
of a retarded/advanced Green's homotopy in the more general sense of 
\cite[Definition~3.5]{GreenHyp}. Such level of generality plays a 
crucial role to ensure uniqueness of retarded/advanced Green's 
homotopies, see \cite[Proposition~3.9]{GreenHyp}. This general and more abstract
concept of a retarded/advanced Green's homotopy is not needed 
for the present paper because a Green's witness $W$ for the complex of 
linear differential operators $(F,Q)$ is given, which allows us to consider 
the explicit choices $\Lambda_\pm$ from \eqref{eqn:Lambdapm}. This considerably simplifies our analysis, 
in particular in view of naturality with respect to $M \in \Loc_m$, 
see Subsection~\ref{subsec:properties} below. 
\end{rem}

We shall now endow the complex $\FFF_\cc(M)[1]\in\Ch_\bbR$ of linear observables
with both a $(-1)$-shifted Poisson structure $\tau_{(-1)}$ and an unshifted one $\tau_{(0)}$. 
Furthermore, we shall construct a symmetric pairing $\tau_D$, called Dirac pairing, that 
trivializes $\tau_{(-1)}$, i.e.\ $\partial \tau_D = \tau_{(-1)}$.
The key ingredients for our construction are the integration pairing 
$(\!(-,-)\!)$ from \eqref{eqn:intpairing}
and the retarded and advanced Green's homotopies $\Lambda_\pm$ from \eqref{eqn:Lambdapm}.
By taking their difference, we define the {\it retarded-minus-advanced cochain map} 
\begin{flalign}\label{eqn:ret-adv-ch}
\Lambda := \Lambda_+ - \Lambda_- : \FFF_\cc(M)[1] \longrightarrow \FFF(M)
\end{flalign}
in $\Ch_\bbR$, where $\Lambda_\pm$ are regarded here as $0$-cochains in 
$[\FFF_\cc(M)[1], \FFF(M)] \in \Ch_\bbR$ (under the isomorphism $[\FFF_\cc(M)[1], \FFF(M)] \cong 
[\FFF_\cc(M), \FFF(M)][-1]$ in $\Ch_\bbR$ given by $(-1)^n$ in degree $n$).
Note that $\Lambda$ is a cochain map because $\partial \Lambda_\pm = j$. 
Similarly, we define the {\it Dirac homotopy} 
\begin{flalign}\label{eqn:Dirac-prop}
\Lambda_D := \frac{1}{2}\,\big(\Lambda_+ + \Lambda_-\big) \in [\FFF_\cc(M)[1], \FFF(M)]^0 
\end{flalign}
as a graded linear map of degree $0$.
We have seen in \eqref{eqn:delLambda} that the cochain map $j: \FFF_\cc(M) \to \FFF(M)$ 
is trivialized by $\Lambda_\pm$. It follows that a similar result 
is achieved by the Dirac homotopy $\Lambda_D$, namely 
\begin{flalign}\label{eqn:delLambdaD}
\partial\Lambda_D = j \in [\FFF_\cc(M)[1], \FFF(M)]^1 \quad.
\end{flalign} 

First, we define the {\it $(-1)$-shifted Poisson structure} 
\begin{flalign}\label{eqn:shiftedPoiss}
\xymatrix@C=5em{
\FFF_\cc(M)[1]^{\otimes 2} \ar@{-->}[rr]^-{\tau_{(-1)}} \ar[d]_-{\cong} && \bbR[1] \\ 
\FFF_\cc(M)[1]\otimes \bbR[1]\otimes \FFF_\cc(M) \ar[r]_-{\gamma \otimes j} & \bbR[1] \otimes \FFF_\cc(M)[1] \otimes \FFF(M) \ar[r]_-{\id \otimes (\!(-,-)\!)[1]} & \bbR[1] \otimes \bbR \ar[u]_-{\cong}
}
\end{flalign}
in $\Ch_\bbR$, where $\gamma$ denotes the symmetric braiding. 
To confirm that \eqref{eqn:shiftedPoiss} defines a $(-1)$-shifted Poisson structure 
we have to check symmetry $\tau_{(-1)} \circ \gamma = \tau_{(-1)}$. 
Indeed, for all homogeneous sections $\psi_1, \psi_2 \in \FFF_\cc(M)[1]$ with compact support, one has 
\begin{flalign}
\tau_{(-1)} \gamma (\psi_1 \otimes \psi_2) &= (-1)^{(\lvert \psi_1 \rvert +1) \lvert \psi_2 \rvert} \int_M{(\psi_2, \psi_1)\, \vol_M} \nn \\
&= (-1)^{\lvert \psi_1 \rvert} \int_M{(\psi_1, \psi_2)\, \vol_M} \nn \\
&= \tau_{(-1)}(\psi_1 \otimes \psi_2) \quad,
\end{flalign}
where in the first and last steps we used the definition of $\tau_{(-1)}$ from \eqref{eqn:shiftedPoiss}
and in the second step we used that the fiber metric $(-,-)$ is graded anti-symmetric, 
see Definition~\ref{defi:CLDO-fibermetric}. 
\sk

Second, we define the {\it unshifted Poisson structure} 
\begin{flalign}\label{eqn:Poiss}
\xymatrix@C=21em{
\FFF_\cc(M)[1]^{\otimes 2} \ar@{-->}[r]^-{\tau_{(0)}} \ar[d]_-{\cong \otimes \Lambda} & \bbR \\ 
\bbR[1] \otimes \FFF_\cc(M) \otimes \FFF(M) \ar[r]_-{\id \otimes (\!(-,-)\!)} & \bbR[1] \otimes \bbR[-1] \ar[u]_-{\cong}
}
\end{flalign}
in $\Ch_\bbR$. To confirm that \eqref{eqn:Poiss} defines an unshifted 
Poisson structure we have to check anti-symmetry 
$\tau_{(0)} \circ \gamma = - \tau_{(0)}$. 
Indeed, for all homogeneous sections $\psi_1, \psi_2 \in \FFF_\cc(M)[1]$ with compact support, one has
\begin{flalign}\label{eqn:Poiss-asym}
\tau_{(0)} \gamma (\psi_1 \otimes \psi_2) &= (-1)^{\lvert \psi_1 \rvert \, \lvert \psi_2 \rvert} \int_M{(\psi_2, GW \psi_1)\,\vol_M} \nn \\
&= - (-1)^{\lvert \psi_1 \rvert} \int_M{(GW \psi_1, \psi_2)\,\vol_M} \nn \\
&= - \int_M (\psi_1, WG \psi_2) \nn \\
&= - \tau_{(0)} (\psi_1 \otimes \psi_2) \quad,
\end{flalign}
where in the first and last steps we used the definition of $\tau_{(0)}$ from \eqref{eqn:Poiss}, 
in the second step we used that the fiber metric $(-,-)$ is 
graded anti-symmetric, see Definition~\ref{defi:CLDO-fibermetric}, 
and in the third step we used Definition~\ref{defi:witness}~(iii) 
and Remark~\ref{rem:propW}~(3). 
\sk

Finally, we define the {\it Dirac pairing} 
\begin{flalign}\label{eqn:pairDirac}
\xymatrix@C=18em{
\FFF_\cc(M)[1]^{\otimes 2} \ar@{-->}[r]^-{\tau_D} \ar[d]_-{\cong \otimes \Lambda_D} & \bbR \\ 
\bbR[1] \otimes \FFF_\cc(M) \otimes \FFF(M) \ar[r]_-{\id \otimes (\!(-,-)\!)} & \bbR[1] \otimes \bbR[-1] \ar[u]_-{\cong}
}
\end{flalign}
as a graded linear map of degree $0$, i.e.\ $\tau_D \in [\FFF_\cc(M)[1]^{\otimes 2}, \bbR]^0$. 
The same calculation as in \eqref{eqn:Poiss-asym} (with the Dirac propagator $G_D$ 
replacing the retarded-minus-advanced one $G$) proves 
symmetry $\tau_D \circ \gamma = \tau_D$. 
Note that $\tau_D$ trivializes the $(-1)$-shifted Poisson structure $\tau_{(-1)}$, i.e.\
\begin{flalign}\label{eqn:partial-tauD}
\partial \tau_D = \tau_{(-1)} \quad.
\end{flalign}
Indeed, for all homogeneous sections $\psi_1,\psi_2 \in \FFF_\cc(M)[1]$ 
with compact support, one has 
\begin{flalign}
\partial \tau_D (\psi_1 \otimes \psi_2) &= \int_M{(Q \psi_1, \Lambda_D \psi_2)\,\vol_M} - (-1)^{\lvert \psi_1 \rvert} \int_M{(\psi_1, \Lambda_D Q_{[1]}\psi_2)\,\vol_M} \nn \\
&= (-1)^{\lvert \psi_1 \rvert} \int_M{(\psi_1, (Q\, \Lambda_D - \Lambda_D\, Q_{[1]}) \psi_2)\,\vol_M} \nn \\
&= (-1)^{\lvert \psi_1 \rvert} \int_M{(\psi_1, \psi_2)\,\vol_M} \nn \\
&= \tau_{(-1)}(\psi_1 \otimes \psi_2) \quad,
\end{flalign}
where in the first step we used the definition of $\tau_D$ from \eqref{eqn:pairDirac}, 
in the second step we used \eqref{eqn:compatibility-pairing}, 
in the third step we used \eqref{eqn:delLambdaD} and in the last step 
we used the definition of $\tau_{(-1)}$ from \eqref{eqn:shiftedPoiss}.

\subsection{\label{subsec:properties}Properties of \texorpdfstring{$\tau_{(-1)}$, $\tau_{(0)}$ and $\tau_D$}{the structures}}
Let us now consider a collection $(F_M, Q_M, (-,-)_M, W_M)_{M \in \Loc_m}$ 
of free BV theories, indexed by $M \in \Loc_m$. 
We assume that $(F_M, Q_M, (-,-)_M, W_M)_{M \in \Loc_m}$ 
is natural with respect to the morphisms $f: M \to N$ in $\Loc_m$ 
in the sense of the next definition. 

\begin{defi}\label{defi:natBV}
A {\it natural collection of free BV theories} 
$(F_M, Q_M, (-,-)_M, W_M)_{M \in \Loc_m}$ consists of 
natural vector bundles $\mathsf{F}^n$, 
natural linear differential operators 
$Q^n: \Gamma(\mathsf{F}^n) \to \Gamma(\mathsf{F}^{n+1})$ and 
$W^n: \Gamma(\mathsf{F}^n) \to \Gamma(\mathsf{F}^{n-1})$ and natural 
fiber metrics $(-,-)^n: \mathsf{F}^n \otimes \mathsf{F}^{1-n} \to \bbR$, 
for all $n \in \bbZ$, such that, for all $M \in \Loc_m$, 
$(F_M, Q_M, (-,-)_M, W_M)$ is a free BV theory in the sense of
Definition~\ref{defi:BV}. 
\end{defi}

The concepts of natural vector bundles, natural fiber metrics and 
natural differential operators, which are relevant 
for the definition above, are recalled in Appendix~\ref{app:natural}. 

\begin{ex}\label{ex:natBV}
Let us upgrade Examples~\ref{ex:P}, \ref{ex:CS} 
and \ref{ex:MW} to natural collections of free BV theories as formalized 
in Definition~\ref{defi:natBV}. 
\sk

Concerning the natural upgrade of Example~\ref{ex:P}, 
it suffices to take as input a natural vector bundle $\mathsf{E}$ 
endowed with a natural fiber metric $\langle-,-\rangle$ and 
a natural linear differential operator $P$ defined on $\mathsf{E}$, 
whose components $P_M$ are formally self-adjoint Green hyperbolic operators, 
for all $M \in \Loc_m$. Taking the natural Green's witness 
given, for all $M \in \Loc_m$, by the identity as in Example~\ref{ex:P}, 
one obtains a natural collection of free BV theories. For instance, 
the natural collection of free BV theories associated with the 
real Klein-Gordon field of mass $m \geq 0$ is obtained 
by taking $\mathsf{E} = \und{\bbR}$ to be 
the natural trivial line bundle, whose components are 
the trivial line bundles $M \times \bbR \to M$, 
endowed with its canonical natural fiber metric $\langle-,-\rangle$, 
given component-wise by the multiplication on $\bbR$, 
and the Klein-Gordon operator $P = \Box + m^2$, whose naturality 
follows from the fact that morphisms of $\Loc_m$ are isometries. 
\sk

The upgrade of Examples \ref{ex:CS} and \ref{ex:MW} to natural 
collections of free BV theories is obtained as follows. 
First, consider the natural vector bundles of differential $k$-forms 
$\Lambda^k$, whose components are the vector bundles $\Lambda^k M\to M$. 
(Naturality follows from the fact that morphisms of $\Loc_m$ 
are open embeddings.) 
Second, note that the wedge product $\wedge$ of differential forms, 
the Hodge operator $\ast$ and the de Rham differential $\dd$ 
are natural with respect to the morphisms in $\Loc_m$. 
(This relies also on the fact that morphisms of $\Loc_m$ 
are orientation-preserving isometries.) This defines
a natural structure in the sense of Definition~\ref{defi:natBV} 
on the collection of free BV theories 
from Examples \ref{ex:CS} and \ref{ex:MW},
which describe linear Chern-Simons theory and Maxwell $p$-forms.
\end{ex}

We summarize below the key facts that will play a crucial role 
in the rest of the paper. These are part 
of Definition \ref{defi:natBV}, or follow from it 
and the constructions outlined in Appendix~\ref{app:natural}, 
especially \eqref{eqn:pull-push}, \eqref{eqn:nat-int-pair}, 
\eqref{eqn:diffop-nat} and \eqref{eqn:diffop-nat-bis}. 
For all $f : M \to N$ in $\Loc_m$, one has the following: 
\begin{enumerate}[label={(\arabic*)}]
\item A pushforward cochain map
$f_\ast : \FFF_\cc(M) \to \FFF_\cc(N)$ in $\Ch_\bbR$ 
for compactly supported sections and a pullback cochain map 
$f^\ast : \FFF(N) \to \FFF(M)$ in $\Ch_\bbR$ for sections. 
(In particular, $Q_N\, f_\ast = f_\ast\, Q_M$ and $Q_M\, f^\ast = f^\ast\, Q_N$.)

\item Naturality of the integration pairing \eqref{eqn:intpairing}, 
i.e.\ the diagram
\begin{flalign}\label{eqn:natsquare-int}
\xymatrix@C=4em{
\FFF_\cc(M) \otimes \FFF(N) \ar[r]^-{\id \otimes f^\ast} \ar[d]_-{f_\ast \otimes \id} & \FFF_\cc(M) \otimes \FFF(M) \ar[d]^-{(\!(-,-)\!)_M}\\
\FFF_\cc(N) \otimes \FFF(N) \ar [r]_-{(\!(-,-)\!)_N}& \bbR[-1] 
}
\end{flalign}
in $\Ch_\bbR$ commutes.

\item Naturality of Green's witnesses, i.e.\ $W_N\, f_\ast = f_\ast\, W_M$ 
and $W_M\, f^\ast = f^\ast\, W_N$.
\end{enumerate}
(1) and (3) entail that also the Green hyperbolic operators $P_M := Q_M\, W_M + W_M\, Q_M$
are natural, hence $P_N\, f_\ast = f_\ast\, P_M$ and 
$P_M\, f^\ast=f^\ast\, P_N$, for all $f : M \to N$ in $\Loc_m$. 
As a consequence of the naturality of $P = (P_M)_{M \in \Loc_m}$ 
and of the theory of Green hyperbolic operators, 
for all $f : M \to N$ in $\Loc_m$, one has the usual naturality property
$f^\ast\, G^N_\pm\, f_\ast = G^M_\pm$ for the retarded/advanced 
Green's operator $G^M_\pm$ associated with $P_M$, see \cite{BG}, 
as well as the analogs $f^\ast\, G^N\, f_\ast = G^M$ and
and $f^\ast\, G^N_D\, f_\ast = G^M_D$ for the 
retarded-minus-advanced propagator $G^M := G^M_+ - G^M_-$ 
and for the Dirac propagator $G^M_D := \tfrac{1}{2}(G^M_+ + G^M_-)$. 
Therefore, the retarded/advanced Green's homotopies
$\Lambda^M_\pm := W_M\, G^M_\pm$, the retarded-minus advanced cochain maps
$\Lambda^M := \Lambda^M_+ - \Lambda^M_-$ and the Dirac homotopies
$\Lambda^M_D := \tfrac{1}{2}(\Lambda^M_+ + \Lambda^M_-)$ inherit the same naturality, 
that is, for all $\Loc_m$-morphisms $f : M \to N$, one has 
\begin{flalign}\label{eqn:nat-Lambdapm}
f^\ast\, \Lambda^N_\pm\, f_\ast = \Lambda^M_\pm \quad, \qquad f^\ast\, \Lambda^N\, f_\ast = \Lambda^M \quad, \qquad f^\ast\, \Lambda^N_D\, f_\ast = \Lambda^M_D \quad.
\end{flalign}
Finally, for all $M \in \Loc_m$, let us consider the $(-1)$-shifted Poisson structures $\tau^M_{(-1)}$ from 
\eqref{eqn:shiftedPoiss}, the unshifted Poisson structures $\tau^M_{(0)}$ from \eqref{eqn:Poiss} 
and the Dirac pairings $\tau^M_D$ from \eqref{eqn:pairDirac}, 
with additional superscripts emphasizing the underlying object in $\Loc_m$. 
As a consequence of \eqref{eqn:nat-Lambdapm} and 
of the naturality of the integration pairing, 
see \eqref{eqn:natsquare-int}, one obtains the following result. 
\begin{lem}
For all $f : M \to N$ in $\Loc_m$, the following holds 
\begin{flalign}\label{eqn:nat-tau}
\tau^N_{(-1)} \circ (f_\ast \otimes f_\ast) = \tau^M_{(-1)} \quad, \qquad \tau^N_{(0)} \circ (f_\ast \otimes f_\ast) = \tau^M_{(0)} \quad, \qquad \tau^N_D \circ (f_\ast \otimes f_\ast) = \tau^M_D \quad.
\end{flalign} 
\end{lem}
\begin{proof}
The first equality follows immediately from \eqref{eqn:shiftedPoiss} 
and \eqref{eqn:natsquare-int}. To prove also the second equality, 
recall \eqref{eqn:Poiss} and, 
for all $\psi_1,\psi_2 \in \FFFF_\cc(M)$, compute 
\begin{flalign}
\tau^N_{(0)}(f_\ast \psi_1 \otimes f_\ast \psi_2) = (\!(\psi_1, f^\ast \Lambda^N f_\ast \psi_2)\!)_M = (\!(\psi_1, \Lambda^M \psi_2)\!)_M = \tau^M_{(0)}(\psi_1 \otimes \psi_2) \quad, 
\end{flalign}
where we used \eqref{eqn:natsquare-int} in the first and 
\eqref{eqn:nat-Lambdapm} in the second step. 
Recalling \eqref{eqn:pairDirac}, the proof of the third equality is the same. 
\end{proof}

This means that $\tau_{(-1)}^M$, $\tau_{(0)}^M$ and $\tau_D^M$ 
are the components at $M \in \Loc_m$ of the natural transformations
$\tau_{(-1)}$, $\tau_{(0)}$ and $\tau_D$, respectively. 
In particular, the assignment to each object $M \in \Loc_m$ 
of the Poisson cochain complex $(\FFF_\cc(M)[1], \tau^M_{(0)}) \in \PCh_\bbR$ 
and to each morphism $f:M\to N$ in $\Loc_m$ of the pushforward 
$f_\ast: (\FFF_\cc(M)[1], \tau^M_{(0)}) \to (\FFF_\cc(N)[1], \tau^N_{(0)})$ in $\PCh_\bbR$
defines a functor $(\FFF_\cc[1],\tau_{(0)}): \Loc_m \to \PCh_\bbR$. 
(Note that \eqref{eqn:nat-tau} expresses the necessary compatibility of $f_\ast$ 
with the unshifted Poisson structures $\tau^M_{(0)}$ and $\tau^N_{(0)}$.)
\sk

The next result shows that classical
analogs of the Einstein causality and time-slice axioms hold. 
To simplify our notation, from now on we shall suppress 
the superscripts and subscripts 
emphasizing the underlying object of $\Loc_m$, 
whenever this information can be inferred from the context.
\begin{theo}\label{th:time-slice-cl}
Let $(F_M, Q_M, (-,-)_M, W_M)_{M \in \Loc_m}$ be a natural 
collection of free BV theories. 
\begin{enumerate}[label={(\alph*)}]
\item For all causally disjoint morphisms 
$f_1: M_1 \to N \leftarrow M_2: f_2$ in $\Loc_m$, 
\begin{flalign}
\tau_{(0)} \circ (f_{1\,\ast} \otimes f_{2\,\ast}) = 0 
\end{flalign}
vanishes.

\item For all Cauchy morphisms $f : M \to N$ in $\Loc_m$, 
the pushforward cochain map
\begin{flalign}\label{eqn:push}
f_\ast : \FFF_\cc (M)[1] \longrightarrow \FFF_\cc (N)[1]
\end{flalign}
in $\Ch_\bbR$ is a quasi-isomorphism.
\end{enumerate}
\end{theo}
\begin{proof}
Item~(a) follows from $J_N(f_1(M_1)) \cap f_2(M_2) = \emptyset$ 
(because $f_1$ and $f_2$ are causally disjoint), the definition of 
the unshifted Poisson structure $\tau_{(0)}$ and the support properties 
of retarded and advanced Green's operators, 
see \eqref{eqn:Poiss} and Definition~\ref{defi:Gpm}. 
\sk
	
To prove also item~(b), we shall construct a quasi-inverse 
$g : \FFF_\cc(N)[1] \to \FFF_\cc(M)[1]$ in $\Ch_\bbR$
for $f_\ast$ and homotopies $\eta \in [\FFF_\cc(N)[1], \FFF_\cc(N)[1]]^{-1}$, witnessing that $f_\ast\, g \sim \id$, 
and $\zeta \in [\FFF_\cc(M)[1], \FFF_\cc(M)[1]]^{-1}$, 
witnessing that $g\, f_\ast \sim \id$. 
Recalling that $f: M \to N$ in $\Loc_m$ 
is by hypothesis a Cauchy morphism, let us consider two spacelike Cauchy surfaces 
$\Sigma_\pm \subset N$ lying inside the image of $f$ 
such that $\Sigma_+ \subset I^+_N(\Sigma_-)$ is contained 
in the chronological future of $\Sigma_-$. Choose a partition of unity
$\{\chi_+, \chi_-\}$ subordinate to the open cover 
$\{I^+_N(\Sigma_-), I^-_N(\Sigma_+)\}$ of $N$. 

\paragraph{Quasi-inverse $g$:} 
We construct a candidate quasi-inverse as the (unique) cochain map
\begin{subequations}\label{eqn:g}
\begin{flalign}
g : \FFF_\cc(N)[1] \longrightarrow \FFF_\cc(M)[1]
\end{flalign}
in $\Ch_\bbR$ that satisfies the equation
\begin{flalign}
j\, f_\ast\, g = \mp \partial(\chi_\pm\, \Lambda): \FFF_\cc(N)[1] \longrightarrow \FFF(N)[1]
\end{flalign}
\end{subequations}
in $\Ch_\bbR$, where $j : \FFF_\cc(N)[1] \to \FFF(N)[1]$ in $\Ch_\bbR$ 
denotes the inclusion forgetting compact supports, the $(-1)$-cochain 
$\chi_\pm \in [\FFF(N),\FFF(N)[1]]^{-1}$ denotes 
multiplication by the partition function $\chi_\pm$ and 
$\partial$ denotes the internal hom differential of 
$[\FFF_\cc(N)[1],\FFF_\cc(N)[1]] \in \Ch_\bbR$.  
Such cochain map $g$ exists (uniquely) because $\mp \partial(\chi_\pm\, \Lambda)$ 
is manifestly a cochain map and, for all 
sections $\psi \in \FFF_\cc(N)[1]$, the section
$-(\partial(\chi_+\, \Lambda)) \psi = (\partial(\chi_-\, \Lambda)) \psi \in \FFF(N)[1]$
lies in the image of the degree-wise injective cochain map $j\, f_\ast$. 
Indeed, the support of the section 
$-(\partial(\chi_+\, \Lambda)) \psi = (\partial(\chi_-\, \Lambda)) \psi$ 
is contained in the compact subset 
$J_N(\supp(\psi)) \cap J^+_N(\Sigma_-) \cap J^-_N(\Sigma_+) \subseteq f(M)$. (The latter subset is compact by \cite[Corollary~A.5.4]{BGP} 
and contained in $f(M)$ because by construction 
$J^+_N(\Sigma_-) \cap J^-_N(\Sigma_+) \subseteq f(M)$.) 

\paragraph{Homotopy $\eta$:}
We construct a candidate homotopy as the (unique) $(-1)$-cochain
\begin{subequations}\label{eqn:eta}
\begin{flalign}
\eta \in [\FFF_\cc(N)[1], \FFF_\cc(N)[1]]^{-1}
\end{flalign}
that satisfies the equation
\begin{flalign}\label{subeqn:tildeeta}
j\, \eta = -\chi_-\, \Lambda_+ - \chi_+\, \Lambda_- \in [\FFF_\cc(N)[1], \FFF(N)[1]]^{-1} \quad,
\end{flalign}
\end{subequations}
where $\Lambda_\pm$ are regarded here as $0$-cochains in 
$[\FFF_\cc(N)[1], \FFF(N)] \in \Ch_\bbR$ 
(under the isomorphism $[\FFF_\cc(N)[1], \FFF(N)] \cong [\FFF_\cc(N), \FFF(N)][-1]$ in $\Ch_\bbR$ given by $(-1)^n$ in degree $n$). 
Such $(-1)$-cochain $\eta$ exists (uniquely) because, for all sections 
$\psi \in \FFF_\cc(N)[1]$, the section 
$\chi_\mp \Lambda_\pm \psi \in \FFF(N)[1]$ lies in the image 
of the degree-wise injective cochain map $j$ that forgets compact supports.
Indeed, the support of the section 
$\chi_\mp \Lambda_\pm \psi \in \FFF(N)[1]$ 
is contained in the compact subset 
$J^\mp_N(\Sigma_\pm) \cap J^\pm_N(\supp(\psi)) \subseteq N$. 
Let us check that $\partial \eta = \id - f_\ast\, g$. 
Since $j$ is degree-wise injective, this follows from 
\begin{flalign}
j\, (\partial\eta) = \partial(- \chi_-\, \Lambda_+ - \chi_+\, \Lambda_-) = j + (\partial \chi_+)\, \Lambda = j\, (\id - f_\ast\, g) \quad,
\end{flalign}
where in the first step we used that $j$ is a cochain map 
and the equation defining $\eta$, in the second step we used 
the Leibniz rule of $\partial$ with respect to the composition, 
$\chi_+ + \chi_- = 1$ (hence $\partial \chi_+ = - \partial \chi_-$), 
$\partial \Lambda_\pm = j$ and $\Lambda = \Lambda_+ - \Lambda_-$ 
and in the last step we used $\partial \Lambda = 0$ and \eqref{eqn:g}.

\paragraph{Homotopy $\zeta$:}
We construct a candidate homotopy as the (unique) $(-1)$-cochain
\begin{subequations}\label{eqn:zeta}
\begin{flalign}
\zeta \in [\FFF_\cc(M)[1], \FFF_\cc(M)[1]]^{-1}
\end{flalign}
that satisfies the equation 
\begin{flalign}
f_\ast\, \zeta = \eta\, f_\ast \in [\FFF_\cc(M)[1], \FFF_\cc(N)[1]]^{-1}\quad.
\end{flalign}
\end{subequations}
Such $(-1)$-cochain $\zeta$ exists (uniquely) because, 
for all homogeneous sections $\psi \in \FFF_\cc(M)[1]$, 
the section $\chi_\mp \Lambda_\pm f_\ast \psi \in \FFF(N)[1]$ 
lies in the image of the degree-wise injective cochain map $f_\ast$. 
Indeed, the support of the section 
$\chi_\mp \Lambda_\pm f_\ast \psi \in \FFF(N)[1]$ 
is contained in the compact subset 
$J^\mp_N(\Sigma_\pm) \cap J^\pm_N(f(\supp(\psi))) \subseteq f(M)$. 
Let us check that $\partial \zeta = \id - g\, f_\ast$. 
Since $f_\ast$ is degree-wise injective, this follows from 
\begin{flalign}
f_\ast\, (\partial\zeta) = (\partial\eta)\, f_\ast = f_\ast\, (\id - g\,f_\ast)
\end{flalign}
where in the first step we used that $f_\ast$ is a cochain map 
and the definition of $\zeta$ and in the last step we used 
$\partial\eta = \id - f_\ast\, g$. 
\end{proof}

To conclude this section, we record a simple result 
relating $\tau_{(0)}$ and $\tau_D$ via time-ordering. 
\begin{propo}\label{propo:time-ord}
Let $(F_M, Q_M, (-,-)_M, W_M)_{M \in \Loc_m}$ be a natural 
collection of free BV theories. 
Then, for all time-ordered pairs $(f_1, f_2) : (M_1,M_2) \to N$
in $\Loc_m$, 
\begin{flalign}\label{eqn:time-ord-cl}
\tau_D \circ (f_{1\,\ast} \otimes f_{2\,\ast}) = \frac{1}{2}\, \tau_{(0)} \circ (f_{1\,\ast} \otimes f_{2\,\ast}) \quad.
\end{flalign}
\end{propo}
\begin{proof}
For all $\psi_1 \in \FFF_\cc(M_1)[1]$
and $\psi_2 \in \FFF_\cc(M_2)[1]$, recalling the support properties 
of retarded and advanced Green's operators from 
Definition~\ref{defi:Gpm}, one computes
\begin{flalign}
\frac{1}{2}\,\tau_{(0)}(f_{1\, \ast} \psi_1 \otimes f_{2\, \ast} \psi_2) &= \frac{1}{2}\, \int_N \left( f_{1\, \ast} \psi_1, \Lambda f_{2\, \ast} \psi_2 \right)_N\, \vol_N \nn \\
&= \frac{1}{2} \, \int_N{(f_{1\, \ast} \psi_1, \Lambda_+ f_{2\, \ast}\psi_2)_N\,\vol_N} \nn \\
&= \int_N{(f_{1\, \ast} \psi_1, \Lambda_D f_{2\, \ast} \psi_2)_N\,\vol_N} \nn \\ 
&= \tau_D(f_{1\, \ast} \psi_1 \otimes f_{2\, \ast} \psi_2) \quad.
\end{flalign}
The first step uses the definition of the unshifted 
Poisson structure $\tau_{(0)}$, see \eqref{eqn:Poiss}. 
Both the second and third steps use that 
$f_1(M_1) \cap J^-_N(f_2(M_2)) = \emptyset$ is empty 
(because $(f_1,f_2)$ is time-ordered), in combination either with 
$\Lambda = \Lambda_+ - \Lambda_-$ or with 
$\Lambda_D =\tfrac{1}{2} (\Lambda_+ + \Lambda_-)$. 
The last step uses the definition of the Dirac pairing $\tau_D$, 
see \eqref{eqn:pairDirac}. 
\end{proof}


\section{\label{sec:Quant}Quantizations and comparison}
In this section we shall present two a priori different approaches to the quantization 
of a natural collection $(F_M, Q_M, (-,-)_M, W_M)_{M \in \Loc_m}$ 
of free BV theories. 
First, in Subsection~\ref{subsec:BV-quant} we shall construct a time-orderable
prefactorization algebra $\F \in \tPFA$ by deforming the ordinary 
differential of the symmetric algebra $\Sym (\FFF_\cc(M)[1])$ 
generated by linear observables 
with the BV Laplacian, as prescribed by the BV formalism \cite{CG, CG2}. 
Second, in Subsection \ref{subsec:star-prod} we shall construct an 
AQFT $\A \in \AQFT$ by deforming the 
commutative multiplication of $\Sym (\FFF_\cc(M)[1]) \in \dgAlg_\bbC$ 
to the non-commutative Moyal-Weyl star product. 
These two constructions involve different input data. 
More specifically, the time-orderable prefactorization algebra 
$\F \in \tPFA$ relies only on the natural $(-1)$-shifted fiber metric $(-,-)$ 
through the natural $(-1)$-shifted Poisson structure $\tau_{(-1)}$ 
(except for the time-slice axiom), 
while the AQFT $\A \in \AQFT$ relies also
on the natural Green's witness $W$ through 
the natural unshifted Poisson structure $\tau_{(0)}$. 
Last, we shall show in Subsection~\ref{subsec:comp} that, when
both $(-,-)$ and $W$ are given, the natural Dirac pairing 
$\tau_D$ leads to an isomorphism 
$T: \F \to \F_\A$ in $\tPFA$ to the time-orderable prefactorization 
algebra $\F_\A \in \tPFA$ canonically associated with $\A \in \AQFT$, see~\cite{FA-vs-AQFT}. 
Let us mention that the deformation parameter $\hbar > 0$ 
will not be formal in our constructions below. 
Indeed, all expansions in powers of $\hbar$ that appear later on 
actually stop at finite order, see for instance the comment after \eqref{eqn:star-prod}. 
\sk

As a preparatory step, let us present a geometric construction 
that will be used frequently in the rest of the paper. 
\begin{lem}\label{lem:factorization}
Let $\und{f} = (f_1,\dots,f_n): \und{M}=(M_1,\dots,M_n) \to N$ 
be a time-ordered tuple in $\Loc_m$ of length $n \geq 2$. 
Then there exist $M \in \Loc_m$, $f: M \to N$ in $\Loc_m$ 
and a time-ordered tuple $\und{f}^\prime = (f^\prime_1,\dots,f^\prime_{n-1}): (M_1,\ldots,M_{n-1}) \to M$ 
in $\Loc_m$ of length $n-1$ such that 
$(f,f_n): (M,M_n) \to N$ is a time-ordered pair in $\Loc_m$ 
and $f \circ f_i^\prime = f_i$, for all $i=1,\ldots,n-1$. In short,
each time-ordered $n$-tuple $\und{f}: \und{M} \to N$, for $n\geq 2$, admits a factorization
\begin{flalign}
\xymatrix{
\ar[dr]_-{(\und{f}^\prime,\id_{M_n})}\und{M} \ar[rr]^-{\und{f}}&& N\\
&(M,M_n)\ar[ru]_-{(f,f_n)}&
}
\end{flalign}
with $\und{f}^\prime$ a time-ordered $(n-1)$-tuple and $(f,f_n)$ a time-ordered pair.
\end{lem}
\begin{proof}
Recalling Subsection~\ref{subsec:Green}, we define the subset 
\begin{flalign}
M := J^{+\cap-}_N \left( \bigcup_{i=1}^{n-1} f_i(M_i) \right) \subseteq N 
\end{flalign}
as the causally convex hull of the union of the images of $f_i$, 
for $i=1,\ldots,n-1$. Since the images are open, $M \subseteq N$
is open and causally convex. 
Endowing it with the restriction of the orientation, time-orientation 
and metric of $N$ defines an object $M \in \Loc_m$ and promotes 
the subset inclusion $M \subseteq N$ to a morphism $f: M \to N$ in $\Loc_m$. 
Since, for each $i=1,\ldots,n-1$, $f_i(M_i) \subseteq M$ by construction, 
$f_i: M_i \to N$ in $\Loc_m$ factors as 
$f_i = f \circ f_i^\prime$, where $f_i^\prime: M_i \to M$ in $\Loc_m$
is the codomain restriction of $f_i$. To conclude, 
let us also check that $(f,f_n)$ is a time-ordered pair, 
i.e.\ $J^+_N(f(M)) \cap f_n(M_n) = \emptyset$. 
By contraposition, suppose that the intersection is not empty. 
Then there exists a future directed causal curve in $N$ 
emanating from $f(M)$ and reaching $f_n(M_n)$. 
Since any point in the causally convex hull $M$ is by definition 
in the causal future of $f_i(M_i)$, for some $i=1,\ldots,n-1$, 
it follows that there exists a future directed causal curve in $N$ 
emanating from $f_i(M_i)$, for some $i=1,\ldots,n-1$, 
and reaching $f_n(M_n)$, leading to a contradiction 
with the hypothesis that $\und{f}$ is time-ordered. 
\end{proof}

\subsection{\label{subsec:BV-quant}BV quantization}
Let $(F_M, Q_M, (-,-)_M, W_M)_{M \in \Loc_m}$ be a natural 
collection of free BV theories. 
Consider the symmetric algebra $\Sym(\FFF_\cc(M)[1]) \in \dgCAlg_\bbC$ 
generated by the complexification of $\FFF_\cc(M)[1] \in \Ch_\bbR$, 
whose differential $\Q$ is defined by the differential 
$Q_{[1]} = -Q$ of $\FFF_\cc(M)[1]$ and the graded Leibniz rule. 
BV quantization consists of deforming $\Q$ 
by means of the {\it BV Laplacian}
\begin{flalign}\label{eqn:BV-Laplacian}
\Delta_\BV := \Delta_{\tau_{(-1)}} \in \big[ \Sym(\FFF_\cc(M)[1]),\Sym(\FFF_\cc(M)[1]) \big]^1 \quad,
\end{flalign}
which is the Laplacian associated to the $(-1)$-shifted Poisson structure $\tau_{(-1)}$, 
see Definition~\ref{defi:Laplacian} and \eqref{eqn:shiftedPoiss}. 
Explicitly, one defines the degree increasing graded linear map 
\begin{flalign}
\Q_\hbar := \Q + \ii \hbar\, \Delta_\BV \in \big[ \Sym(\FFF_\cc(M)[1]),\Sym(\FFF_\cc(M)[1]) \big]^1 \quad,
\end{flalign}
where $\hbar > 0$ is Planck's constant and $\ii\in\bbC$ is the imaginary unit. 
Note that $\Q_\hbar$ defines a new differential since it squares to zero
\begin{flalign}
\Q_\hbar^2 = \Q^2 + \ii \hbar\, \partial \Delta_\BV - \hbar^2\, \Delta_\BV^2 = 0 \quad,
\end{flalign}
where we used $\Q^2 = 0$, $\partial \Delta_\BV = \Delta_{\partial \tau_{(-1)}} = 0$ 
and $\Delta_\BV^2 = - \Delta_\BV^2 = 0$, 
see \eqref{eqn:partial-Laplacian} and \eqref{eqn:LaplacianLaplacian}. 
We define the {\it cochain complex of quantum observables} 
\begin{flalign}\label{eqn:qBV-complex}
\F(M) := \big( \Sym(\FFF_\cc(M)[1]), \Q_\hbar \big) \in \Ch_\bbC 
\end{flalign}
by replacing the original differential $\Q$ with the quantized one $\Q_\hbar$. 
The assignment $\Loc_m \ni M \mapsto \F(M) \in \Ch_\bbC$ 
of the cochain complex of quantum observables can be promoted 
to a time-orderable prefactorization algebra $\F \in \tPFA$. 
For this purpose, we need to define time-ordered products 
that are compatible with the quantized differential $\Q_\hbar$. 
This is the goal of the next proposition. 
\begin{propo}\label{propo:t-ord-prod}
Let $\und{f}: \und{M} \to N$ be a time-orderable tuple in $\Loc_m$ 
of length $n$. Then the time-ordered product 
\begin{flalign}\label{eqn:BV-product}
\xymatrix@C=4em{
\F(\und{M}) \ar@{-->}[rr]^-{\F(\und{f})} \ar[rd]_-{\bigotimes_i f_{i\,\ast}} && \F(N) \\
& \F(N)^{\otimes n} \ar[ru]_-{\mu^{(n)}}
}\quad,
\end{flalign}
is a cochain map, i.e.\ 
$\Q_\hbar\, \F(\und{f}) = \F(\und{f})\, \Q_{\hbar\, \otimes}$. 
Here $f_{i\,\ast}$ denotes the symmetric algebra extension 
of the pushforward cochain map 
$f_{i\,\ast}: \FFF_\cc(M_i)[1] \to \FFF_\cc(N)[1]$
for compactly supported sections, 
see Subsection~\ref{subsec:properties}, and $\mu^{(n)}$ denotes the 
$n$-ary multiplication on the symmetric algebra $\Sym(\FFF_\cc(N)[1]) \in \dgCAlg_\bbC$. 
\end{propo}
\begin{proof}
Since $\Q$ is natural and compatible $\Q\, \mu = \mu\, \Q_{\otimes}$ 
with the symmetric algebra multiplication $\mu$, 
one has $\Q\, \F(\und{f}) = \F(\und{f})\, \Q_{\otimes}$. 
Hence, it suffices to prove the analog 
$\Delta_\BV\, \F(\und{f}) = \F(\und{f})\, \Delta_{\BV\,{\otimes}}$ 
for the BV Laplacian $\Delta_\BV$. 
Furthermore, since the symmetric algebra multiplication $\mu$ 
is commutative, it suffices to prove the claim for $\und{f}$ time-ordered. 
We argue by induction on the length $n$. 
For $n=0$, the time-ordered product $\bbC \to \F(N)$ defined above 
assigns the unit $\mu^{(0)} = \oone$ of the symmetric algebra, 
hence the claim follows from $\Delta_\BV(\oone) = 0$, 
see Definition~\ref{defi:Laplacian}. 
For $n=1$, $\F(f) = f_\ast$, hence the claim follows 
because the BV Laplacian $\Delta_\BV$ inherits the naturality 
of the $(-1)$-shifted Poisson structure $\tau_{(-1)}$, 
see \eqref{eqn:f-Laplacian} and \eqref{eqn:nat-tau}. 
For $n=2$, one computes 
\begin{flalign}
\Delta_\BV\circ \F(f_1,f_2) &= \Delta_\BV\circ \mu\circ (f_{1\,\ast} \otimes f_{2\,\ast}) \nn \\ 
&= \mu\circ \left( \Delta_{\BV\, \otimes}\, + \pair{-,-}_{(-1)} \right)\circ (f_{1\,\ast} \otimes f_{2\,\ast}) \nn \\
&= \mu\circ (f_{1\,\ast} \otimes f_{2\,\ast})\circ \Delta_{\BV\, \otimes} \nn \\
&= \F(f_1,f_2)\circ \Delta_{\BV\,\otimes} \quad,
\end{flalign}
where in the first and last steps we used the definition 
of the time-ordered product $\F(f_1,f_2)$, 
in the second step we used the degree increasing graded endomorphism 
$\pair{-,-}_{(-1)} := \pair{-,-}_{\tau_{(-1)}}$ 
to spell out the modified Leibniz rule of $\Delta_\BV$, 
see Definitions~\ref{defi:endo} and \ref{defi:Laplacian}, and 
in the third step we used naturality of the BV Laplacian $\Delta_\BV$ 
and that $\pair{-,-}_{(-1)}$ vanishes on the image of 
$f_{1\,\ast} \otimes f_{2\,\ast}$, because
$f_1(M_1) \cap f_2(M_2) = \emptyset$ and $\tau_{(-1)}$ 
vanishes on sections with disjoint supports, 
see \eqref{eqn:shiftedPoiss}. 
For $n \geq 3$, taking $M \in \Loc_m$, $f: M \to N$ in $\Loc_m$ 
and a time-ordered tuple $\und{f}^\prime: (M_1,\ldots,M_{n-1}) \to M$ 
in $\Loc_m$ as provided by Lemma~\ref{lem:factorization}, one computes 
\begin{flalign}
\nn \F(\und{f}) &= \mu^{(n)}\circ \bigotimes_{i=1}^n f_{i\, \ast} 
= \mu\circ (f_\ast \otimes f_{n\, \ast}) \circ \left( \left( \mu^{(n-1)}\circ \bigotimes_{i=1}^{n-1} f_{i\, \ast}^\prime\right) \otimes \id \right) \\
&= \F(f,f_n)\circ (\F(\und{f}^\prime) \otimes \id) \quad, 
\end{flalign}
where in the first and last steps we used the definition 
of the time-ordered product $\F(\und{f})$ and in the second step 
we used $\mu^{(n)} = \mu\circ (\mu^{(n-1)} \otimes \id)$, 
$f \circ f^\prime_i = f_i$, for all $i=1,\ldots,n-1$, 
and the naturality of the symmetric algebra multiplication $\mu$. 
Hence, the claim for length $n \geq 3$ follows from lengths $2$ and $n-1$. 
\end{proof}

With these preparations, we define the
time-orderable prefactorization algebra $\F \in \tPFA$ 
by the data listed below:
\begin{enumerate}[label=(\alph*)]
\item For each $M \in \Loc_m$, the cochain complex 
$\F(M) \in \Ch_\bbC$ from \eqref{eqn:qBV-complex};

\item For each time-orderable tuple $\und{f} : \und{M} \to N$ in $\Loc_m$, 
the time-ordered product $\F(\und{f}): \F(\und{M}) \to \F(N)$ 
in $\Ch_\bbC$ from Proposition \ref{propo:t-ord-prod}. 
\end{enumerate}
Note that these data satisfy the axioms of Definition~\ref{defi:tPFA} 
because $\FFF_\cc[1]: \Loc_m \to \Ch_\bbR$ is a functor, 
see Subsection~\ref{subsec:properties}, and the symmetric algebra multiplication $\mu$ 
is associative, unital and commutative. 
The resulting time-orderable prefactorization algebra $\F \in \tPFA$ 
satisfies the time-slice axiom, as explained by the next proposition. 
\begin{propo}\label{propo:tPFA-time-slice}
If $f : M \to N$ in $\Loc_m$ is a Cauchy morphism, then 
$\F(f): \F(M) \to \F(N)$ in $\Ch_\bbC$ is a quasi-isomorphism.
\end{propo}
\begin{proof}
For any $L \in \Loc_m$, consider the filtration of
$\F(L) = \bigoplus_{n \geq 0} \Sym^n(\FFF_\cc(L)[1])\in\Ch_\bbC$ 
associated with symmetric powers. Explicitly, we denote 
the subcomplex of $\F(L)$ consisting of symmetric powers up to $p \geq 0$ by 
\begin{flalign}\label{eqn:filtration}
F_p(\F(L)) := \left( \bigoplus_{n = 0}^p \Sym^n(\FFF_\cc(L)[1]), \Q_\hbar \right) \subseteq \F(L) \quad.
\end{flalign}
(Note that this filtration is compatible with the quantized differential 
$\Q_\hbar = \Q + \ii \hbar\, \Delta_\BV$ because the original 
differential $\Q$ preserves the symmetric power and the BV Laplacian 
$\Delta_\BV$ lowers the symmetric power by $2$.)
The resulting filtration is bounded from below, i.e.\ $F_p(\F(L)) = 0$
vanishes, for all $p < 0$. The quotient maps 
$\F(L) \to \F(L)/F_p(\F(L))$ in $\Ch_\bbC$, for all $p \in \bbZ$, 
form a universal cone, i.e.\ $\F(L) \cong \lim_{p \in \bbZ} \F(L)/F_p(\F(L))$. 
This shows that the filtration is complete, see~\cite{Eilenberg-Moore}. 
Furthermore, for $p \geq 0$, the $p$-th component of the associated graded cochain complex
\begin{flalign}\label{eqn:associated-graded}
E^\circ_p(L) := F_p(\F(L))/F_{p-1}(\F(L)) \cong \Sym^p(\FFF_\cc(L)[1]) \in \Ch_\bbC 
\end{flalign}
is isomorphic to the $p$-th symmetric power of $\FFF_\cc(L)[1] \in \Ch_\bbC$ 
(endowed with the original differential $\Q$) because 
the BV Laplacian $\Delta_\BV$ lowers the symmetric power by $2$, 
see~\eqref{eqn:Laplacian-explicit}. Functoriality with respect 
to $L \in \Loc_m$ of the filtration~\eqref{eqn:filtration} and 
naturality of the isomorphism~\eqref{eqn:associated-graded} entail that, 
for all $f: M \to N$ in $\Loc_m$, the diagram 
\begin{flalign}
\xymatrix{
E^\circ_p(M) \ar[r]^-{E^\circ_p(f_\ast)} \ar[d]_-{\simeq} & E^\circ_p(N) \ar[d]^-{\simeq} \\
\Sym^p(\FFF_\cc(M)[1]) \ar[r]_-{f_\ast} & \Sym^p(\FFF_\cc(N)[1])
}
\end{flalign}
in $\Ch_\bbC$ commutes. Since the bottom cochain map is a 
quasi-isomorphism by Theorem~\ref{th:time-slice-cl}, 
the claim follows from \cite[Theorem~7.4]{Eilenberg-Moore}. 
\end{proof}

\begin{ex}\label{ex:tPFA}
Taking the natural free BV theories from Example \ref{ex:natBV} 
as inputs, the constructions and results from this subsection 
produce time-orderable prefactorization algebras 
satisfying the time-slice axiom that quantize ordinary field theories, 
linear Chern-Simons theory and Maxwell $p$-forms 
(including linear Yang-Mills theory for $p=1$). 
\end{ex}

\subsection{\label{subsec:star-prod}Moyal-Weyl star product}
Let $(F_M, Q_M, (-,-)_M, W_M)_{M \in \Loc_m}$ be a natural 
collection of free BV theories and consider again 
the symmetric algebra $\Sym(\FFF_\cc(M)[1]) \in \dgCAlg_\bbC$ 
generated by the complexification of $\FFF_\cc(M)[1] \in \Ch_\bbR$.
Canonical quantization can be realized by deforming the commutative multiplication $\mu$ 
of the symmetric algebra $\Sym(\FFF_\cc(M)[1]) \in \dgCAlg_\bbC$ to the {\it Moyal-Weyl star product}
\begin{flalign}\label{eqn:star-prod}
\xymatrix@C=4em{
\Sym(\FFF_\cc(M)[1])^{\otimes 2} \ar@{-->}[rr]^-{\mu_\hbar} \ar[dr]_-{\exp \left( \frac{\ii \hbar}{2}\, \pair{-,-}_{(0)} \right) ~~~~~~~~~~} && \Sym(\FFF_\cc(M)[1]) \\
& \Sym(\FFF_\cc(M)[1])^{\otimes 2} \ar[ur]_-{\mu}
}\quad,
\end{flalign}
where $\pair{-,-}_{(0)} := \pair{-,-}_{\tau_{(0)}}$ 
is the degree preserving graded endomorphism associated with 
the unshifted Poisson structure $\tau_{(0)}$, see 
Definition~\ref{defi:endo} and \eqref{eqn:Poiss}. Note that, for all 
polynomials $a,b \in \Sym(\FFF_\cc(M)[1])$, the exponential 
series defining $\mu_\hbar(a \otimes b)$ truncates to a finite sum. 
In particular, there is no need to regard $\hbar$ as a formal parameter. 
\begin{rem}
The Moyal-Weyl star product $\mu_\hbar$ is a non-commutative deformation 
of the commutative multiplication $\mu$ of the symmetric algebra 
$\Sym(\FFF_\cc(M)[1]) \in \dgCAlg_\bbC$ in the sense that the multiplications 
\begin{flalign}
\mu_\hbar = \mu + \mathcal{O}(\hbar)
\end{flalign} 
coincide up to terms of order at least $\hbar$ and, moreover, 
the $\mu_\hbar$-commutator 
\begin{flalign}
[-,-]_\hbar  = \ii \hbar\, \{-,-\}_{(0)} + \mathcal{O}(\hbar^2)
\end{flalign}
is proportional to the Poisson bracket 
$\{-,-\}_{(0)} := \mu \circ \pair{-,-}_{(0)}$, see Remark~\ref{rem:Poiss-bracket}, 
up to terms of order at least $\hbar^2$. 
\end{rem}

The Moyal-Weyl star product $\mu_\hbar$ is manifestly a degree 
preserving graded linear map. Furthermore, it is associative and unital 
with respect to $\oone \in \Sym(\FFF_\cc(M)[1])$ as a consequence 
of the properties of the degree preserving graded endomorphism 
$\pair{-,-}_{(0)} = \pair{-,-}_{\tau_{(0)}}$ and of the exponential. 
Let us also check that the Moyal-Weyl star product $\mu_\hbar$ is 
compatible with the differential $\Q$ of $\Sym(\FFF_\cc(M)[1])$, i.e.\
\begin{flalign}
\partial \mu_\hbar &= \mu\circ  \partial \exp \left( \frac{\ii \hbar}{2} \pair{-,-}_{(0)} \right) \nn \\
&= \mu\circ \Bigg(\sum_{n \geq 1} \frac{1}{n!} \left( \frac{\ii \hbar}{2} \right)^{n} \,\sum_{k=0}^{n-1} \pair{-,-}_{(0)}^k\circ \big(\partial \pair{-,-}_{(0)}\big)\circ \pair{-,-}_{(0)}^{n-1-k}\Bigg) \nn \\
&= 0\quad,
\end{flalign}
where in the first step we used the compatibility $\partial \mu = 0$ 
of the symmetric algebra multiplication $\mu$ with the differential $\Q$, 
in the second step we expanded the exponential series 
and applied the Leibniz rule for $\partial$ and in the last step 
we used that $\partial \pair{-,-}_{(0)} = \pair{-,-}_{\partial \tau_{(0)}} = 0$ vanishes, 
see \eqref{eqn:partial-endo} and recall that $\tau_{(0)}$ 
is a cochain map. 
Therefore, we define the quantized differential graded algebra 
\begin{flalign}
\A(M) := \big( \Sym(\FFF_\cc(M)[1]), \mu_\hbar, \oone \big) \in \dgAlg_\bbC\quad.
\end{flalign}
To promote the assignment $\Loc_m \ni M \mapsto \A(M) \in \dgAlg_\bbC$ 
of the quantized differential graded algebra to a functor, 
we check the naturality of the Moyal-Weyl star product $\mu_\hbar$ 
with respect to morphisms $f: M \to N$ in $\Loc_m$, i.e.\
\begin{flalign}
f_\ast \circ \mu_\hbar = \mu \circ (f_\ast \otimes f_\ast) \circ \exp \left( \frac{\ii \hbar}{2} \pair{-,-}_{(0)} \right) = \mu_\hbar \circ (f_\ast \otimes f_\ast) \quad,
\end{flalign}
where in the first step we used naturality of the 
symmetric algebra multiplication $\mu$ and in the second step 
we used the naturality of the unshifted Poisson structure $\tau_{(0)}$, 
see \eqref{eqn:nat-tau}, in combination with \eqref{eqn:f-endo} 
at all orders in $\hbar$. We are now ready to define the functor 
\begin{flalign}\label{eqn:AAA}
\A : \Loc_m \longrightarrow \dgAlg_\bbC
\end{flalign}
that assigns to any object $M \in \Loc_m$ 
the differential graded algebra $\A(M) \in \dgAlg_\bbC$ 
and to any morphism $f : M \to N$ in $\Loc_m$ the morphism
$\A(f): \A(M) \to \A(N)$ in $\dgAlg_\bbC$, whose underlying cochain map 
is the symmetric algebra extension of the pushforward cochain maps 
$f_\ast : \FFF_\cc(M)[1] \to \FFF_\cc(N)[1]$ in $\Ch_\bbR$. The next 
proposition shows that $\A$ is an AQFT.
\begin{propo}\label{propo:AQFT}
The functor $\A : \Loc_m \to \dgAlg_\bbC$ from \eqref{eqn:AAA} 
satisfies the Einstein causality and time-slice axioms of 
Definition~\ref{defi:AQFT}, hence $\A \in \AQFT$ is an AQFT.
\end{propo}
\begin{proof}
First, let us check the Einstein causality axiom. 
For causally disjoint morphisms $f_1 : M_1 \to N \gets M_2 : f_2$ in $\Loc_m$, 
Definition~\ref{defi:endo} applied to the unshifted Poisson structure 
$\tau_{(0)}$ and Theorem~\ref{th:time-slice-cl} entail that
\begin{flalign}
\pair{-,-}_{(0)} \circ \left( f_{1\, \ast} \otimes f_{2\, \ast} \right) = 0 
\end{flalign}
vanishes. Therefore, on the image of $ f_{1\, \ast} \otimes f_{2\, \ast}$ 
the Moyal-Weyl star product $\mu_\hbar$, see \eqref{eqn:star-prod}, coincides 
\begin{flalign}
\mu_\hbar \circ \left( f_{1\, \ast} \otimes f_{2\, \ast} \right) = \mu \circ \big( f_{1\,\ast} \otimes f_{2\,\ast} \big)
\end{flalign}
with the symmetric algebra multiplication $\mu$. Since the latter 
is commutative, the Einstein causality axiom follows.
\sk

Second, let us check the time-slice axiom. Given a Cauchy morphism 
$f : M \to N$ in $\Loc_m$, it suffices to show that the cochain map 
underlying $\A(f): \A(M) \to \A(N)$ in $\dgAlg_\bbC$ is a 
quasi-isomorphism. This is the case because 
the cochain map underlying $\A(f)$ is by definition the symmetric 
algebra extension of the pushforward cochain map 
$f_\ast : \FFF_\cc(M)[1] \to \FFF_\cc(N)[1]$ in $\Ch_\bbR$, 
which is a quasi-isomorphism by Theorem~\ref{th:time-slice-cl}. 
\end{proof}

\begin{ex}\label{ex:AQFT}
Taking the natural free BV theories from Example \ref{ex:natBV} 
as inputs, the constructions and results of this subsection 
produce AQFTs that quantize ordinary field theories, 
linear Chern-Simons theory and Maxwell $p$-forms 
(including linear Yang-Mills theory for $p=1$). 
In the case of the Klein-Gordon field and of Maxwell $p$-forms, earlier 
constructions of the same AQFTs can be found in \cite{linYM,hrep}. 
\end{ex}

\subsection{\label{subsec:comp}Comparison}
This subsection compares the two different quantization schemes from
Sections~\ref{subsec:BV-quant}~and~\ref{subsec:star-prod}. 
More specifically, we establish an isomorphism between the 
time-orderable prefactorization algebra $\F \in \tPFA$ 
constructed using the BV formalism and 
the time-orderable prefactorization algebra $\F_\A \in \tPFA$
canonically associated with the AQFT $\A \in \AQFT$. 
For this purpose, recall from \cite{FA-vs-AQFT} that 
$\F_\A \in \tPFA$ consists of the data listed below: 
\begin{enumerate}[label=(\alph*)]
\item For each $M \in \Loc_m$, the cochain complex $\F_\A(M) := \Sym(\FFF_\cc(M)[1]) \in \Ch_\bbC$ underlying $\A(M) \in \dgAlg_\bbC$;
\item For each time-orderable tuple $\und{f} : \und{M} \to N$ 
in $\Loc_m$, the time-ordered product 
\begin{flalign}\label{eqn:t-ord-prod}
\xymatrix@C=4em{
\F_\A(\und{M}) \ar@{-->}[rr]^-{\F_\A(\und{f})} \ar[rd]_-{\bigotimes_i f_{i\, \ast}} && \F_\A(N) \\
& \F_\A(N)^{\otimes n} \ar[ru]_-{\mu_\hbar^{(\rho)}}
}
\end{flalign}
in $\Ch_\bbC$, where $n$ denotes the length of the tuple $\und{f}$, 
$\rho$ is a time-ordering permutation for $\und{f}$ and
$\mu_\hbar^{(\rho)} := \mu_\hbar^{(n)} \circ \gamma_\rho$ denotes 
the $n$-ary Moyal-Weyl star product in the order prescribed by $\rho$. 
(The Einstein causality axiom of $\A$ ensures that 
$\F_\A(\und{f})$ does not depend on the choice of 
the time-ordering permutation $\rho$.) 
\end{enumerate}
The above data fulfill the axioms of Definition~\ref{defi:tPFA}, 
see \cite{FA-vs-AQFT} for more details. 
\sk

In preparation for our comparison result stated in Theorem~\ref{th:iso}, 
the next lemma explains how the time-orderable prefactorization algebra $\F_\A \in \tPFA$ 
captures the usual time-ordered products built out of the {\it Dirac multiplication} 
\begin{flalign}\label{eqn:Dirac-mult}
\xymatrix@C=4em{
\Sym(\FFF_\cc(M)[1])^{\otimes 2} \ar@{-->}[rr]^-{\mu_D} \ar[dr]_-{\exp \left( \ii \hbar\, \pair{-,-}_{D} \right) ~~~~~~~~~~} && \Sym(\FFF_\cc(M)[1]) \\
& \Sym(\FFF_\cc(M)[1])^{\otimes 2} \ar[ur]_-{\mu}
}\quad,
\end{flalign}
where $\pair{-,-}_{D} := \pair{-,-}_{\tau_D}$ denotes 
the degree preserving graded endomorphism associated with the natural 
Dirac pairing $\tau_D$ \eqref{eqn:pairDirac}. 
Note that the Dirac multiplication $\mu_D$ is associative, 
unital with respect to $\oone \in \Sym(\FFF_\cc(M)[1])$ 
and commutative because the Dirac pairing 
$\tau_D$ is symmetric; however, it is not compatible with 
the differential $\Q$ of $\F_\A(M) \in \Ch_\bbC$ because $\partial \tau_D = \tau_{(-1)}$ does 
not vanish, see \eqref{eqn:partial-endo} and \eqref{eqn:partial-tauD}. 
Furthermore, the naturality of $\tau_D$ and that of the symmetric algebra 
multiplication $\mu$ entail that the Dirac multiplication $\mu_D$ is natural too. 
\begin{lem}\label{lem:t-ord-prod-alt}
Let $\und{f}: \und{M} \to N$ be a time-orderable tuple in $\Loc_m$ 
of length $n$. Then the time-ordered product $\F_\A(\und{f})$ 
can be computed using the Dirac multiplication $\mu_D$, i.e.\
\begin{flalign}
\xymatrix@C=4em{
\F_\A(\und{M}) \ar[rr]^-{\F_\A(\und{f})} \ar[dr]_{\bigotimes_i f_{i\, \ast}} && \F_\A(N) \\
& \F_\A(N)^{\otimes n} \ar[ur]_-{\mu_D^{(n)}} &
}\quad.
\end{flalign}
\end{lem}
\begin{proof}
Since the Dirac multiplication $\mu_D$ is commutative 
and $\F_\A \in \tPFA$ is a time-orderable factorization algebra 
(hence its time-ordered products are equivariant with respect to 
permutations, see Definition~\ref{defi:tPFA}), 
it suffices to check the claim for $\und{f}$ time-ordered. 
We argue by induction on the length $n$. For $n=0$ and $n=1$, the claim holds 
because $\mu_D^{(0)} = \oone = \mu_\hbar^{(0)}$ and $\mu_D^{(1)} = \id = \mu_\hbar^{(1)}$. 
For $n=2$, Proposition~\ref{propo:time-ord} entails that, 
for all $k \geq 1$, 
\begin{flalign}\label{eqn:pairD=pair0}
\pair{-,-}_D^k \circ (f_{1\, \ast} \otimes f_{2\, \ast}) = \left( \frac{1}{2} \pair{-,-}_{(0)} \right)^k \circ (f_{1\, \ast} \otimes f_{2\, \ast}) \quad. 
\end{flalign}
Then one computes 
\begin{flalign}
\mu_D \circ (f_{1\, \ast} \otimes f_{2\, \ast}) &= \mu \circ \exp \left( \ii \hbar\, \pair{-,-}_D \right) \circ (f_{1\, \ast} \otimes f_{2\, \ast}) \nn \\ 
&= \mu \circ \exp \left( \frac{\ii \hbar}{2}\, \pair{-,-}_{(0)} \right) \circ (f_{1\, \ast} \otimes f_{2\, \ast}) \nn \\
&= \F_\A(f_1,f_2) \quad,
\end{flalign}
where in the first step we used the definition of the Dirac 
multiplication $\mu_D$, see \eqref{eqn:Dirac-mult}, 
in the second step we used \eqref{eqn:pairD=pair0} 
and in the last step we used the definition of the time-ordered product 
$\F_\A(f_1,f_2)$, see \eqref{eqn:t-ord-prod}. 
For $n \geq 3$, taking $M \in \Loc_m$, $f: M \to N$ in $\Loc_m$ 
and a time-ordered tuple $\und{f}^\prime: (M_1,\ldots,M_{n-1}) \to M$ 
in $\Loc_m$ as provided by Lemma~\ref{lem:factorization}, one computes 
\begin{flalign}
\mu_D^{(n)} \circ \bigotimes_{i=1}^n f_{i\, \ast} &= \mu_D \circ (f_\ast \otimes f_{n\, \ast}) \circ \left( \left( \mu_D^{(n-1)} \circ \bigotimes_{i=1}^{n-1} f_{i\, \ast}^\prime \right) \otimes \id \right) \quad, 
\end{flalign}
where we used $\mu_D^{(n)} = \mu_D\circ (\mu_D^{(n-1)} \otimes \id)$, 
$f \circ f^\prime_i = f_i$, for all $i=1,\ldots,n-1$, 
and the naturality of the Dirac multiplication $\mu_D$. 
Hence, the claim for length $n \geq 3$ follows from lengths $2$ and $n-1$. 
\end{proof}

The alternative description from Lemma~\ref{lem:t-ord-prod-alt} 
of the time-ordered products of $\F_\A \in \tPFA$ plays a key role 
in the proof of our main result. 
\begin{theo}\label{th:iso}
The time-orderable prefactorization algebra $\F \in \tPFA$ 
(constructed via the BV formalism in Subsection~\ref{subsec:BV-quant}) 
and the time-orderable prefactorization algebra $\F_\A \in \tPFA$ 
associated with the AQFT $\A \in \AQFT$ 
(constructed via the Moyal-Weyl star product 
in Subsection~\ref{subsec:star-prod}) are isomorphic. 
Explicitly, the {\it time-ordering map}
\begin{flalign}
T:= \exp(\ii  \hbar\, \Delta_D) : \F \overset{\cong}{\longrightarrow} \F_\A
\end{flalign}
in $\tPFA$ is an isomorphism. Here $\Delta_D := \Delta_{\tau_D}$, called Dirac Laplacian, 
is the Laplacian associated with the Dirac pairing $\tau_D$, 
see Definition~\ref{defi:Laplacian} and \eqref{eqn:pairDirac}. 
\end{theo}
\begin{proof}
Suppose that $T$ as defined above is a morphism of time-orderable
prefactorization algebras. Then it is also an isomorphism with inverse 
$T^{-1} := \exp(-\ii \hbar\, \Delta_D) : \F_\A \to \F$ in $\tPFA$. 
Therefore, it suffices to check that $T$ is a morphism of time-orderable 
prefactorization algebras. We split this check in two parts. 
First, we show the compatibility with differentials, i.e.\ that, 
for all $M \in \Loc_m$, the $M$-component 
$T_M: \F(M) \to \F_\A(M)$ is a cochain map. 
Second, we show the compatibility with time-ordered products, 
i.e.\ that, for all time-orderable tuples $\und{f}: \und{M} \to N$ 
in $\Loc_m$, $T_N \circ \F(\und{f}) = \F_\A(\und{f}) \circ T_{\und{M}}$.

\paragraph{Compatibility with differentials:}
Recall the BV and Dirac Laplacians  
$\Delta_\BV := \Delta_{\tau_{(-1)}}\in [\Sym(\FFF_\cc(M)[1]),\Sym(\FFF_\cc(M)[1])]^1$ and 
$\Delta_D := \Delta_{\tau_D} \in [\Sym(\FFF_\cc(M)[1]),\Sym(\FFF_\cc(M)[1])]^0$, 
see Definition~\ref{defi:Laplacian}, \eqref{eqn:shiftedPoiss} 
and \eqref{eqn:pairDirac}. From \eqref{eqn:partial-Laplacian}, 
\eqref{eqn:LaplacianLaplacian} and \eqref{eqn:partial-tauD} it follows that 
\begin{flalign}\label{eqn:BV-D-Laplacians}
\partial \Delta_D = \Delta_{\partial \tau_{D}} = \Delta_{\BV} \quad, \qquad \Delta_\BV\circ \Delta_D = \Delta_D\circ \Delta_\BV \quad.
\end{flalign}
Therefore, regarding $T_M$ as a $0$-cochain in 
$[\Sym(\FFF_\cc(M)[1]),\Sym(\FFF_\cc(M)[1])] \in \Ch_\bbC$, one computes 
\begin{flalign}\label{eqn:partialT}
\Q\circ T_M - T_M\circ \Q = \partial T_M = \sum_{n \geq 1} \frac{1}{n!}\, \partial \big( (\ii \hbar\, \Delta_D)^n \big)
= T_M\circ (\ii \hbar\, \Delta_\BV)\quad.
\end{flalign}
In the first step we used the definition of $\partial$. 
In the second step we expanded the exponential that defines $T_M$ 
(recall that the series evaluated on any $a \in \Sym(\FFF_\cc(M)[1])$ 
truncates to a finite sum) and used that $\partial$ is linear and 
vanishes on $\id$. In the last step we used the Leibniz rule for $\partial$ 
with respect to composition, \eqref{eqn:BV-D-Laplacians} 
and the definition of $T_M$. Equation~\eqref{eqn:partialT} means that 
$\Q\circ T_M = T_M\circ \Q_\hbar$, which shows that $T_M: \F(M) \to \F_\A(M)$ 
is a cochain map.

\paragraph{Compatibility with time-ordered products:}
Since $\F,\F_\A \in \tPFA$ are time-orderable prefactorization algebras, 
their time-ordered products are equivariant with respect 
to permutations, see Definition~\ref{defi:tPFA}. 
Hence, it suffices to show the compatibility of $T$ 
with time-ordered products for $\und{f}: \und{M} \to N$ time-ordered. 
We argue by induction on the length $n$ of $\und{f}$. 
For $n=0$ this is trivial.
For $n=1$, the claim follows from naturality of the Dirac Laplacian 
$\Delta_D \circ f_\ast = f_\ast \circ \Delta_D$, 
see \eqref{eqn:f-Laplacian} and \eqref{eqn:nat-tau}. 
For $n=2$, one computes
\begin{flalign}
T_N \circ \F(f_1,f_2) &= \mu_D \circ (T_N \otimes T_N) \circ \big( \F(f_1) \otimes \F(f_2) \big) = \F_\A(f_1,f_2) \circ (T_{M_1} \otimes T_{M_2}) \quad,
\end{flalign}
where in the first step we used
$T_N \circ \mu = \mu_D \circ (T_N \otimes T_N)$,
which follows from \eqref{eqn:Deltan-mu}, and  
in the last step we used the claim for $n=1$ 
and Lemma~\ref{lem:t-ord-prod-alt}. 
For $n \geq 3$, taking $M \in \Loc_m$, $f: M \to N$ in $\Loc_m$ 
and a time-ordered tuple $\und{f}^\prime: (M_1,\ldots,M_{n-1}) \to M$ 
in $\Loc_m$ as provided by Lemma~\ref{lem:factorization}, one computes 
\begin{flalign}
T_N \circ \F(\und{f}) = T_N \circ \F(f,f_n) \circ \big( \F(\und{f}^\prime) \otimes \id \big) = \F_\A(f,f_n) \circ \big( \F_\A(\und{f}^\prime) \otimes \id \big) \circ T_{\und{M}} = \F_\A(\und{f}) \circ T_{\und{M}} \quad,
\end{flalign}
where in the first and last steps we used the composition and identity 
axioms of time-orderable prefactorization algebras 
and $f \circ f_i^\prime = f_i$, for all $i=1,\ldots,n-1$, and 
in the second step we used the claim for lengths $2$ and $n-1$.
\end{proof}

\begin{ex}
The abstract comparison result established in 
Theorem~\ref{th:iso}, when specialized to the time-orderable 
prefactorization algebras from Example~\ref{ex:tPFA} 
and the corresponding AQFTs from Example~\ref{ex:AQFT}, 
provides concrete comparison results between
the time-orderable prefactorization algebras 
and the AQFTs quantizing ordinary field theories, 
linear Chern-Simons theory and Maxwell $p$-forms 
(including linear Yang-Mills theory for $p=1$). 
Our result generalizes the earlier comparison result in \cite{GwilliamRejzner},
which is formulated only for ordinary field theories, to the case of 
gauge and also higher gauge theories.
\end{ex}


\section*{Acknowledgments}
The work of M.B.\ and G.M.\ is supported in part by the MIUR Excellence 
Department Project awarded to Dipartimento di Matematica, 
Università di Genova (CUP D33C23001110001) and it is fostered by 
the National Group of Mathematical Physics (GNFM-INdAM (IT)). 
G.M.\ is supported by a PhD scholarship of the University of Genova (IT).
A.S.\ gratefully acknowledges the support of 
the Royal Society (UK) through a Royal Society University 
Research Fellowship (URF\textbackslash R\textbackslash 211015)
and Enhancement Awards (RGF\textbackslash EA\textbackslash 180270, 
RGF\textbackslash EA\textbackslash 201051 and RF\textbackslash ERE\textbackslash 210053).

\section*{Conflict of interest statement}
On behalf of all authors, the corresponding author states that there is no conflict of interest. 

\section*{Data availability statement}
All data generated or analyzed during this study are contained in this document.


\appendix

\section{\label{app:natural}Natural geometric structures}
This appendix provides an explicit description of the constituents 
of the natural collections of free BV theories 
from Definition \ref{defi:natBV}, namely natural vector bundles, 
natural fiber metrics and natural linear differential operators. 
This appendix also outlines the relevant constructions that lead 
to the key facts (1--3) from Subsection \ref{subsec:properties}. 
\sk

Let us consider the category $\VBun_\bbR$, whose objects are pairs $(M,E)$ 
consisting of $M \in \Loc_m$ and a finite rank real vector bundle 
$E \to M$ and whose morphisms are pairs 
$(f,\bar{f}): (M_1,E_1) \to (M_2,E_2)$ consisting of a morphism 
$f: M_1 \to M_2$ in $\Loc_m$ and a vector bundle map $\bar{f}: E_1 \to E_2$ 
over $f$ that acts as an isomorphism on the fibers, namely such that, 
for all $p \in M$, the linear map $\bar{f}_p: E_{1\,p} \to E_{2\, f(p)}$ 
is an isomorphism. 
Let us consider the evident functor 
$\pi: \VBun_\bbR \to \Loc_m$, $(M,E) \mapsto M$. 
\sk

A {\it natural vector bundle} $\mathsf{E}$ is a section of $\pi$, 
i.e.\ a functor $\mathsf{E}: \Loc_m \to \VBun_\bbR$ 
such that $\pi \circ \mathsf{E} = \id$. This means that
$\mathsf{E}$ sends $M \in \Loc_m$ to an object of the form 
$\mathsf{E}(M) = (M,E_M) \in \VBun_\bbR$ and $f: M_1 \to M_2$ 
to a morphism of the form 
$\mathsf{E}(f) = (f,E_f): \mathsf{E}(M_1) \to \mathsf{E}(M_2)$ in $\VBun_\bbR$. 
Given a natural vector bundle $\mathsf{E}$ and a morphism $f:M_1 \to M_2$ 
in $\Loc_m$, one constructs the pullback and pushforward linear maps 
\begin{flalign}\label{eqn:pull-push}
f^\ast: \Gamma(E_{M_2}) \longrightarrow \Gamma(E_{M_1}) \quad, \qquad f_\ast: \Gamma_\cc(E_{M_1}) \longrightarrow \Gamma_\cc(E_{M_2})
\end{flalign} 
for sections and sections with compact support, respectively. 
Explicitly, given $\varphi \in \Gamma(E_{M_2})$, 
$f^\ast \varphi \in \Gamma(E_{M_1})$ is defined by 
$f^\ast \varphi := E_f^{-1} \circ \varphi \circ f$, where we note that 
$E_f^{-1}$ can be inverted as only fibers over $f(M_1)$ are involved. 
Furthermore, given $\psi \in \Gamma_\cc(E_{M_1})$, 
$f_\ast \psi \in \Gamma_\cc(E_{M_2})$ is defined as the extension 
by zero along the open embedding $f(M_1) \subseteq M_2$ 
of the compactly supported section 
$E_f \circ \psi \circ f^{-1}: f(M_1) \to E_{M_2}$. 
Note that $f^\ast$ and $f_\ast$ upgrade the assignments 
$M \mapsto \Gamma(E_M)$ and $M \mapsto \Gamma_\cc(E_M)$ to functors 
\begin{flalign}
\Gamma(\mathsf{E}): \Loc_m^\op \longrightarrow \Vec_\bbR\quad, \qquad \Gamma_\cc(\mathsf{E}): \Loc_m \longrightarrow \Vec_\bbR
\end{flalign}
with values in the category of real vector spaces and linear maps. 
By constructions of $f^\ast$ and $f_\ast$, it follows that 
\begin{flalign}\label{eqn:pullpush=id}
f^\ast\, f_\ast\, \psi = \psi \quad, \qquad f_\ast\, f^\ast\, \varphi = \varphi \quad, 
\end{flalign}
for all $\psi \in \Gamma_\cc(E_{M_1})$ and 
all $\varphi \in \Gamma_\cc(E_{M_2})$ with $\supp(\varphi) \subseteq f(M_1)$. 
\sk

A {\it natural fiber metric} $\langle-,-\rangle$ on a natural vector bundle 
$\mathsf{E}$ is a natural transformation 
$\langle-,-\rangle: \mathsf{E} \otimes \mathsf{E} \to \underline{\bbR}$ 
to the natural trivial line bundle $\underline{\bbR}: \Loc_m\to \VBun_\bbR\,,~M\mapsto M\times\bbR$,
whose components $\langle-,-\rangle_M: E_M \otimes E_M \to M\times \bbR$ 
are fiber metrics, for all $M \in \Loc_m$. 
Given $M \in \Loc_m$, one defines the integration pairing 
\begin{subequations}
\begin{flalign}
\langle\!\langle -, - \rangle\!\rangle_M: \Gamma_\cc(E_M) \otimes \Gamma(E_M) \longrightarrow \bbR 
\end{flalign}
by
\begin{flalign}
\langle\!\langle \psi, \varphi \rangle\!\rangle_M := \int_M{\langle \psi, \varphi \rangle_M\,\vol_M} \quad,
\end{flalign}
\end{subequations}
for all $\psi \in \Gamma_\cc(E_M)$ and $\varphi \in \Gamma(E_M)$, 
see also~\eqref{eqn:fibintpairing}. Because $\langle-,-\rangle$ 
is a natural fiber metric, it follows that the integration pairing 
$\langle\!\langle -, - \rangle\!\rangle$ is natural in the following sense: 
for all $f: M_1 \to M_2$ in $\Loc_m$, the diagram 
\begin{flalign}\label{eqn:nat-int-pair}
\xymatrix@C=3.5em{
\Gamma_\cc(E_{M_1}) \otimes \Gamma(E_{M_2}) \ar[r]^-{\id \otimes f^\ast} \ar[d]_-{f_\ast \otimes \id} & \Gamma_\cc(E_{M_1}) \otimes \Gamma(E_{M_1}) \ar[d]^-{\langle\!\langle -, - \rangle\!\rangle_{M_1}} \\ 
\Gamma_\cc(E_{M_2}) \otimes \Gamma(E_{M_2}) \ar[r]_-{\langle\!\langle -, - \rangle\!\rangle_{M_2}} & \bbR 
}
\end{flalign}
in $\Vec_\bbR$ commutes. Indeed, for all $\psi \in \Gamma_\cc(E_{M_1})$ 
and $\varphi \in \Gamma(E_{M_2})$, one has 
\begin{flalign}
\langle\!\langle \psi, f^\ast \varphi \rangle\!\rangle_{M_1} = \langle\!\langle f^\ast f_\ast \psi, f^\ast \varphi \rangle\!\rangle_{M_1} = \int_{M_1} \big( \langle f_\ast \psi, \varphi \rangle_{M_2} \circ f \big)\, \vol_{M_1} = \langle\!\langle f_\ast \psi, \varphi \rangle\!\rangle_{M_2} \quad, 
\end{flalign}
where we used \eqref{eqn:pullpush=id} in the first step, 
naturality of $\langle-,-\rangle$ and the definitions~\eqref{eqn:pull-push} 
of $f^\ast$ and $f_\ast$ in the second step and naturality of the integral 
with respect to $\Loc_m$-morphisms in the last step. 
\sk

A {\it natural linear differential operator} $P$ 
between natural vector bundles $\mathsf{E}$ and $\mathsf{E}^\prime$ 
is a natural transformation 
$P: \Gamma(\mathsf{E}) \to \Gamma(\mathsf{E}^\prime)$ 
whose components $P_M: \Gamma(E_M) \to \Gamma(E^\prime_M)$ 
are linear differential operators, for all $M \in \Loc_m$. 
Explicitly, this means that, for all $f: M_1 \to M_2$ in $\Loc_m$, 
the diagram 
\begin{flalign}\label{eqn:diffop-nat}
\xymatrix@C=3.5em{
\Gamma(E_{M_2}) \ar[r]^-{P_{M_2}} \ar[d]_-{f^\ast} & \Gamma(E^\prime_{M_2}) \ar[d]^-{f^\ast} \\ 
\Gamma(E_{M_1}) \ar[r]_-{P_{M_1}} & \Gamma(E^\prime_{M_1})
}
\end{flalign}
in $\Vec_\bbR$ commutes. The latter diagram entails 
that $P$ defines also a natural transformation 
$P: \Gamma_\cc(\mathsf{E}) \to \Gamma_\cc(\mathsf{E}^\prime)$. 
Indeed, for all $f: M_1 \to M_2$ in $\Loc_m$, the diagram 
\begin{flalign}\label{eqn:diffop-nat-bis}
\xymatrix@C=3.5em{
\Gamma_\cc(E_{M_1}) \ar[r]^-{P_{M_1}} \ar[d]_-{f_\ast} & \Gamma_\cc(E^\prime_{M_1}) \ar[d]^-{f_\ast} \\ 
\Gamma_\cc(E_{M_2}) \ar[r]_-{P_{M_2}} & \Gamma_\cc(E^\prime_{M_2})
}
\end{flalign}
in $\Vec_\bbR$ commutes, as it follows from the straightforward computation 
\begin{flalign}
f_\ast\, P_{M_1} = f_\ast\, P_{M_1}\, f^\ast\, f_\ast = f_\ast\, f^\ast\, P_{M_2}\, f_\ast = P_{M_2}\, f_\ast \quad, 
\end{flalign}
where we used \eqref{eqn:pullpush=id} in the first and last steps 
and \eqref{eqn:diffop-nat} in the second step.


\end{document}